\def\FULLVER{}

\documentclass{ecai}
\usepackage{graphicx}
\usepackage{latexsym}

\usepackage[utf8]{inputenc}
\usepackage{amsmath,amssymb}
\newtheorem{lemma}[theorem]{Lemma}
\newtheorem*{theorem*}{Theorem}
\newtheorem*{lemma*}{Lemma}
\newtheorem{corollary}[theorem]{Corollary}
\newtheorem{observation}[theorem]{Observation}
\theoremstyle{definition}

\usepackage{algcompatible}
\usepackage{algorithm}

\usepackage{mathtools}
\usepackage{mathrsfs}
\DeclareMathAlphabet{\mathpzc}{OT1}{pzc}{m}{it}

\newcommand*{\RR}{%
  \textsf{I\kern-.3ex R}%
}
\usepackage[dvipsnames]{xcolor}

\usepackage{blindtext}
\usepackage{hyperref}

\usepackage{csquotes}
\usepackage{array}

\usepackage[normalem]{ulem}

\ifdefined\DRAFT
	\newcommand{\er}[1]{\textcolor{blue}{#1}}
	\newcommand{\erel}[1]{\er{(Erel says: #1)}}
	\newcommand{\eden}[1]{\textcolor{red}{(Eden says: #1)}}
	
	\newcommand{\yon}[1]{\textcolor{Green}{(Yonatan says: #1)}}
     
\else
	\newcommand{\er}[1]{#1}
	\newcommand{\erel}[1]{}
	\newcommand{\eden}[1]{}
	
	\newcommand{\yon}[1]{}
\fi

\newcommand{\multApprox}{\alpha}
\newcommand{\additiveApprox}{\epsilon}
\newcommand{\DEFmultApprox}{\alpha}
\newcommand{\DEFadditiveApprox}{\epsilon}
\newcommand{\DEFmultError}{\overline{\multApprox}}
\newcommand{\DEFmultErrorOf}[1]{\overline{#1}}
\newcommand{\DEFadditiveError}{\epsilon}

\newcommand{\Hquad}{\hspace{0.5em}}

\newcommand{\multError}{\beta}
\newcommand{\additiveError}{\epsilon}

\newcommand{\leximinPreferred}{\succ}
\newcommand{\nLeximinPreferred}{\nsucc}

\newcommand{\alphaBetaPreferred}{\succ_{(\DEFmultApprox,\DEFadditiveApprox)}}
\newcommand{\alphaBetaPreferredParams}[2]{\succ_{(#1,#2)}}
\newcommand{\nAlphaBetaPreferred}{\nsucc_{(\DEFmultApprox,\DEFadditiveApprox)}}
\newcommand{\nAlphaBetaPreferredParams}[2]{\nsucc_{(#1,#2)}}

\newcommand{\relationSetAlphaBeta}{\mathcal{R}_{(\multApprox,\additiveApprox)}}
\newcommand{\relationSetParams}[2]{\mathcal{R}_{(#1,#2)}}

\newcommand{\sortedValues}[1]{\mathbf{V}^{\uparrow}(#1)}

 \newcommand{\valBy}[2]{\mathpzc{V}_{#1}^{\uparrow}(#2)}

\newcommand{\retSol}{x^*}
\newcommand{\ztVar}[1]{z_t}

\newcommand{\lexmaxmin} {\operatorname{lex} \max \min}

\usepackage{diagbox}
\usepackage{multirow, tabularray}
\usepackage{rotating}
\usepackage{booktabs}

\newcommand{\progBasic}{P1}
\newcommand{\progSums}{P2}
\newcommand{\progCompact}{P2-Comp}
\newcommand{\progLinear}{P3}

\usepackage[justification=centering, font=footnotesize]{caption}

\ifdefined\ONECOL
	\newcommand{\displayComsoc}[1]{#1}
        \newcommand{\displayEcai}[1]{}
        \newcommand{\forXinY}[2]{ #1 = 1, \ldots, #2}
        \newcommand{\equWithExp}[2]{#2 && (#1)}
\else
	\newcommand{\displayComsoc}[1]{}
    \newcommand{\displayEcai}[1]{#1}
    \newcommand{\forXinY}[2]{\forall #1 \in [#2]}
    \newcommand{\equWithExp}[2]{#1\text{:}\\ & #2}
\fi

\ifdefined\FULLVER
    \newcommand{\fullVer}{}
    \newcommand{\appendixName}[2]{#1}
\else
    \newcommand{\fullVer}{ of the full version}
    \newcommand{\appendixName}[2]{#2}
\fi

\ecaisubmission   

\begin{document}

\begin{frontmatter}

\title{Leximin Approximation: From Single-Objective to Multi-Objective}

\author[A]{\fnms{Eden}~\snm{Hartman}\orcid{0000-0001-5819-4432}\thanks{
Contact information: eden.r.hartman@gmail.com, avinatan@cs.biu.ac.il, aumann@cs.biu.ac.il,  erelsgl@gmail.com.}}
\author[A,B]{\fnms{Avinatan}~\snm{Hassidim}\orcid{0000-0002-3855-344X}}
\author[A]{\fnms{Yonatan}~\snm{Aumann}\orcid{0000-0002-6217-671X}}\author[C]{\fnms{Erel}~\snm{Segal-Halevi}\orcid{0000-0002-7497-5834}} 

\address[A]{Bar-Ilan University, Israel} 
\address[B]{Google, Israel} 
\address[C]{Ariel University, Israel} 


\begin{abstract}
Leximin is a common approach to multi-objective optimization,
frequently employed in fair division applications.
In leximin optimization, one first aims to maximize the smallest objective value; subject to this, one maximizes the second-smallest objective; and so on. 
Often, even the single-objective problem of maximizing the smallest value cannot be solved accurately.  
What can we hope to accomplish for leximin optimization in this situation?
Recently, Henzinger et al. (2022) defined a notion of \emph{approximate} leximin optimality.
Their definition, however, considers only an additive approximation.   

In this work, we first define the notion of approximate leximin optimality, allowing both multiplicative and additive errors. 
We then show how to compute, in polynomial time, such an approximate leximin solution, using an oracle that finds an approximation to a single-objective problem. The approximation factors of the algorithms are closely related: an $(\multApprox,\additiveApprox)$-approximation for the single-objective problem (where $\multApprox \in (0,1]$ and $\additiveApprox \geq 0$ are the multiplicative and additive factors respectively) translates into an $\left(\frac{\multApprox^2}{1-\multApprox + \multApprox^2}, \frac{\additiveApprox}{1-\multApprox +\multApprox^2}\right)$-approximation for the multi-objective leximin problem, regardless of the number of objectives. 

\end{abstract}

\end{frontmatter}


\eden{To myself: to remove page numbering before submission}

\section{Introduction}
Many real life scenarios involve more than one objective. 
These situations are often modeled as \emph{multi-objective optimization problems}. These are defined by the set of possible decisions, along with functions that describe the different objectives.
As a concrete example, we use the context of social choice, in which the objective functions represent people's utilities.
Different criteria can be used to determine optimality when considering multi-objectives. 
For example, the \emph{utilitarian} criterium aims to maximize the sum of utilities, whereas the \emph{egalitarian} criterium aims to maximize the least utility.
This paper focuses on the \emph{leximin} criterium, according to which one aims to maximize the least utility, and, subject to this, maximize the second-smallest utility, and so on.
In the context of social choice, 
the leximin criterium is considered a criterium of fairness,
as strives to benefit, as much as possible, the least fortunate in society.

Common algorithms for finding a leximin optimal solution are iterative, optimizing one or more single-objective optimization problems at each iteration (for example~\cite{airiau_portioning_2019,arbiv_fair_2022,bei_truthful_2022, nace_max-min_2008,Ogryczak2004TelecommunicationsND,willson}).
Often, these single-objective problems cannot be solved exactly (e.g. when they are computationally hard, or when there are numeric inaccuracies in the solver), but can be solved approximately.
In this work, we define an approximate variant of leximin and show how such an approximation can be computed, given approximate single-objective solvers.

\paragraph{The Challenge}
When single-objective solvers only \emph{approximate} the optimal value, existing methods for extending the solvers to leximin optimally may fail, as we illustrate next.

A common algorithm,
independently proposed many times, e.g.~\cite{airiau_portioning_2019,nace_max-min_2008,Ogryczak2004TelecommunicationsND, willson}, is based on the notion of \textit{saturation}. It operates roughly  as follows.
In the first iteration, the algorithm looks for the maximum value that all objective functions can achieve simultaneously, denoted $z_1$. Then, it determines which of the objective functions are saturated --- that is, cannot achieve more than $z_1$ given that the others achieve at least $z_1$.
Afterwards, in each iteration $t$, given that for any $i<t$ the objective-functions that were determined saturated in the $i$'th iteration achieve at least $z_i$, it looks for the maximum value that all other objective-functions can achieve simultaneously, denoted $z_t$, and then determines which of those functions are saturated.
When all functions become saturated, the algorithm ends.

  

The following simple example  demonstrates the problem that may arise when the individual solver may return sub-optimal results. Consider the following problem:
\begin{align*}
    \lexmaxmin \quad &\{f_1(x) = x_1, f_2(x)=x_2\} \\
    s.t. \quad  &(1) \Hquad x_1 + x_2 \leq 1, \quad (2) \Hquad x \in \mathbb{R}^2_{+}
\end{align*}
As $f_1$ and $f_2$ are symmetric, the leximin optimal solution in this case is $(0.5,0.5)$.
Now suppose that rather than finding the exact value $0.5$, the solver returns the value $0.49$. 
The optimal value of $f_1$ given that $f_2$ achieves at least $0.49$ is $0.51$, and vice versa for $f_2$.
As a consequence, none of the objective functions would be determined saturated, and the algorithm may not terminate. 
One could perhaps define an objective as ''saturated'' if its maximum attainable value is close to the maximum $z_t$, but there is no guarantee that this would lead to a good approximation\footnote{An example will be given on request.}
%


\paragraph{Contributions}
This paper studies the problem of leximin optimization in multi-objective optimization problems, 
focusing on problems for which even the single-objective problems cannot be solved exactly in polynomial time.
\eden{Was: 'Our contribution is threefold.', should it be rewritten differently for two? maybe: 'We present two key contributions in this research.'}
Our contribution is twofold.

First, a new definition of leximin approximation is presented.
It captures both multiplicative and additive errors.
The definition has several desirable properties, including that a leximin optimal solution is always approximately-optimal,
and that the definition is equivalent to the original one in the absence of errors.

Second, an algorithm is provided that, given an approximation algorithm for a single-objective problem, computes a leximin approximation to the multi-objective problem.
The algorithm was first presented by Ogryczak
and {\'{S}}liwi{\'{n}}ski~\cite{Ogryczak_2006} for exact leximin-optimization. In contrast to the saturation-based algorithm described in the Introduction, this algorithm always terminates, even when the single-objective solver is inaccurate. 
Moreover, the accuracy of the returned solution is closely correlated with the accuracy of the single-objective solver --- given an $(\multApprox,\additiveApprox)$-approximation algorithm for the single-objective problem (where $\multApprox$ and $\additiveApprox$ describe the multiplicative and additive factors respectively), the returned solution is an $\left(\frac{\multApprox^2}{1-\multApprox + \multApprox^2}, \frac{\additiveApprox}{1-\multApprox +\multApprox^2}\right)$-approximation of leximin.
Importantly, this holds for any number of objectives.

Our algorithm is applicable whenever a single-objective solver is inaccurate due to numerical issues, such as floating-point rounding errors, or stopping an iterative solution procedure before complete convergence.
%
%

\paragraph{Organization} 
Section \ref{sec:preliminaries} gives preliminary knowledge and basic definitions. 
Section \ref{sec:approx-leximin-def} presents the definition of leximin approximation.
An algorithm for computing such an approximation is presented in Section \ref{sec:algo-short}.
Section \ref{sec:future} concludes with some future work directions.

\subsection{Related Work}
This paper is related to a large body of research, which can be classified into three main fields: multi-objective optimization problems, approximation variants of known solution concepts, and algorithms for finding optimal leximin solutions.

In general multi-objective\footnote{Multi-objective is also called \emph{multi-criteria} (for example in~\cite{ehrgott_multicriteria_2005}).} optimization problems, finding a leximin\footnote{Leximin is also called \emph{Max-Min fairness} (for example in~\cite{nace_max-min_2008}), \emph{Lexicographic Min-Max} (for example in~\cite{Ogryczak_2006}), \emph{Lexicographic max-ordering} (for example in~\cite{ehrgott_multicriteria_2005}) and \emph{Leximax} (for example in~\cite{henzinger_et_al_FORC_2022}).} optimal solution is quite common goal~\cite{ehrgott_multicriteria_2005}, which is still an open challenge.
Studies on this topic are usually focused on a specific problem and leverages its special characteristics --- the structure of the \textit{feasible region} and \textit{objective-functions} that describe the concrete problem at hand.
In this paper, we focus on the widely studied domain of resource allocation problems~\cite{moulin2004fair}.
In that context, as leximin maximization is an extension of egalitarian welfare maximization, it is usually mentioned when fairness is desired.

There are cases where a leximin optimal solution can be calculated in polynomial time, for example in: fair allocation of divisible items~\cite{willson}, giveaway lotteries~\cite{arbiv_fair_2022}, portioning with ordinal preferences~\cite{ airiau_portioning_2019}, cake sharing~\cite{bei_truthful_2022}, multi-commodity ﬂow networks~\cite{nace_max-min_2008}, and location problems~\cite{Ogryczak_1997_loc}.
However, even when algorithms are theoretically polynomial, they can still be inaccurate in practice, for example due to numeric round-off errors.

In other cases, calculating a leximin optimal solution is NP-hard, for example in: representative
cohort selection~\cite{henzinger_et_al_FORC_2022}, fair combinatorial auctions~\cite{bouveret_computing_2009}, package upgrade-ability~\cite{cabral2022sat}, and allocating papers to referees~\cite{garg_assigning_2010, lian_conference_2018}.
However, to our knowledge, 
studies of this kind typically suggest non-polynomial algorithms and heuristics for solving small instances of the general problem and empirically evaluate their efficiency, rather than suggesting polynomial-time approximation algorithms.

Another approach to leximin optimization is to represent the leximin ordering by an aggregation function. 
Such a function takes a utility vector and returns a number, such that a solution is leximin-preferred over another if-and-only-if its aggregate number is higher.
Finding such a function will of course reduce the problem to solving only one single-objective optimization problem.
Unfortunately, it is known that no aggregate function can represent the leximin ordering in \emph{all} problems~\cite{moulin2004fair,Ogryczak2004TelecommunicationsND}.
Still, there are interesting cases in which such functions can be found. 
For example, Yager~\cite{yager_ordered_1988} suggested that the \emph{ordered weighted averaging (OWA)} technique can be used when there is a lower bound on the difference between any two possible utilities. 
However, it is unclear how (and whether) approximating the aggregate function would translate to approximating leximin.
%
%

To the best of our knowledge, other general approximations of leximin exist but they are less common.
They are usually mentioned in the context of robustness or noise (e.g.~\cite{henzinger_et_al_FORC_2022,kalai_lexicographic_2012}) and lack characteristics that we emphasize within the context of errors.
Most similar to our work is the recent paper by Henzinger et al.~\cite{henzinger_et_al_FORC_2022}.
This paper presents several approximation variants of leximin for the case of \emph{additive} errors in the single-objective problems.
Their motivation is different than ours --- they use approximation as a method to improve efficiency and ensure robustness to noise. However, one of their definitions,  (\emph{$\epsilon$-tradeoff Leximax}) fits our motivation of achieving the best possible leximin-approximation in the presence of errors.
In fact, our approximation definition can be viewed as a generalization of their definition to include both multiplicative and additive errors.
It should also be noted that the authors mention multiplicative approximation in the their Future Work Section.

\section{Preliminaries}\label{sec:preliminaries}

We denote the set $\{1,\dots, n\}$ by $[n]$ for $n \in \mathbb{N}$.

 \paragraph{Single-objective optimization}

A \emph{single-objective maximization (resp. minimization) problem} is a tuple $(S,f)$ where $S$ is the set of all feasible solutions to the problem (usually $S \subseteq \mathbb{R}^m$ for some $m \in \mathbb{N}$) and $f \colon S \to \mathbb{R}_{\geq 0}$ is a function describing the objective value of a solution $x \in S$.
The goal in such problems is to find an \emph{optimal} solution, that is, a feasible solution $x^* \in S$ that has the maximum (resp. minimum) objective value, that is $f(x^*) \geq f(x)$ (resp. $f(x^*) \leq f(x)$) for any other solution $x \in S$.

A \emph{$(\multApprox,\additiveApprox)$-approximation algorithm} for a single-objective \emph{maximization} problem $(S,f)$ is one that returns a solution $x\in S$, which \emph{approximates} the optimal solution $x^*$ from below. That is, $f(x) \geq \multApprox \cdot f(x^*) - \additiveApprox$ for $\multApprox  \in (0,1]$ and $\additiveApprox \geq 0$ (that describe the multiplicative and additive approximation factors respectively).


A \emph{$p$-randomized} approximation algorithm, for $p\in(0,1]$, is one that returns a solution $x\in S$ such that, with probability $p$, the objective value $f(x)$ is approximately-optimal.

The term "with high probability" (w.h.p.) is used when the success probability is at least $\left(1-\frac{1}{poly(I)}\right)$ where $I$ describes the input size.

\paragraph{Multi-objective optimization}
A \textit{multi-objective maximization} problem~\cite{branke_multiobjective_2008} can be described as follows:
\begin{align*}
    \max \quad &\{f_1(x), f_2(x), \dots f_n(x)\} \\
    s.t. \quad  & x \in S
\end{align*}
Where $S\subseteq \mathbb{R}^m$ for some $m \in \mathbb{N}$ is the \textit{feasible region} and $f_1, f_2, \dots, f_n$ are $n$ \textit{objective-functions} $f_i\colon S \to \mathbb{R}_{\geq 0}$\footnote{Note that the number of objectives in multi-objective optimization is commonly assumed to be constant. However, in this paper, we use a more general setting in which the number of objectives is a parameter of the problem.}.\eden{I think we should say that $f_i\colon S' \to \mathbb{R}_{\geq 0}$ for some  $S \subseteq S'$}
An example application is group decision making:
some $n$ people have to decide on an issue that affects all of them.
The set of possible decisions is $S$, and the utility each person $i$ derives from a decision $x\in S$ is $f_i(x)$.

\paragraph{Ordered outcomes notation}
The $j$'th smallest objective value obtained by a solution $x\in S$ is denoted by $\valBy{j}{x}$, i.e.,
\begin{align*}
    \valBy{1}{x} \leq \valBy{2}{x} \leq \dots \leq \valBy{n}{x}.
\end{align*}
The corresponding \emph{sorted} utility vector is denoted by $\sortedValues{x} = (\valBy{1}{x}, \ldots, \valBy{n}{x})$.
%
%
%

\paragraph{The leximin order}
A solution $y$ is considered \emph{leximin-preferred} over a solution $x$, denoted $y \leximinPreferred x$, if there exists an integer $1 \leq k\leq n $ such that the smallest $(k-1)$ objective values of both are equal, whereas the $k$'th smallest objective value of $y$ is higher:
\begin{align*}
    \forall j < k \colon \quad &\valBy{j}{y} = \valBy{j}{x}\\
    & \valBy{k}{y} > \valBy{k}{x}
\end{align*} 
Two solutions, $x,y,$ are \emph{leximin equivalent} if $\sortedValues{x} = \sortedValues{y}$. 
The leximin order
is a \emph{total} order, and strict between any two solutions that yield \emph{different} sorted utility vectors ($\sortedValues{x} \neq \sortedValues{y}$).  
A \emph{maximum} element of the leximin order is a solution
over which \emph{no} solution is preferred (including solutions that yield the same utilities).

\paragraph{Leximin optimal}
%
A \emph{leximin optimal} solution is a maximum element of the leximin order.
Given a feasible region $S$, as the order is determined only by the utilities, we denote this optimization problem as follows. 
\begin{align*}
	\lexmaxmin \quad &\{f_1(x), f_2(x), \dots f_n(x)\} \\
	s.t. \quad  & x \in S
\end{align*}

\section{Approximate Leximin Optimality}\label{sec:approx-leximin-def}
In this section, we present our definition of leximin approximation in the presence of multiplicative and additive errors, in the context of multi-objective optimization problems.

\subsection{Motivation: Unsatisfactory  Definitions}

Which solutions should be considered approximately-optimal in terms of  leximin? 
Several definitions appear intuitive at first glance.
As an example, suppose we are interested in approximations with an allowable multiplicative error of $0.1$.
Denote the utilities in the leximin-optimal solution by $(u_1,\ldots,u_n)$.
A first potential definition is that any solution in which the sorted utility vector is at least $(0.9\cdot u_1,\ldots,0.9\cdot u_n)$ should be considered approximately-optimal.
For example, if the utilities in the optimal solution are $(1,2,3)$, then a solution with utilities $(0.9, 1.8, 2.7)$ is approximately-optimal.
However, allowing the smallest utility to take the value $0.9$ may substantially increase the maximum possible value of the second (and third) smallest utility --- e.g.~a solution that yields utilities $(0.9, 1000,1000)$ might exist. In that case, a solution with utilities  $(0.9, 1.8, 2.7)$ is very far from optimal.
We expect a good approximation notion to consider the fact that an error in one utility might change the optimal value of the others.

The following, second attempt at a definition, captures this requirement.
An approximately-optimal solution is one that yields utilities at least $(0.9\cdot m_1, 0.9 \cdot m_2, \dots, 0.9 \cdot m_n)$, where $m_1$ is the maximum value of the smallest utility, $m_2$ is the maximum value of the second-smallest utility \emph{among all solutions whose smallest utility is at least $0.9 \cdot m_1$};
$m_3$ is the maximum value of the third-smallest utility among all solutions whose smallest utility is at least $0.9 \cdot m_1$ and their second-smallest utility is at least $0.9\cdot m_2$; and so on. 
In the above example, to be considered approximately-optimal, the smallest utility should be at least $0.9$ and the second-smallest should be at least $900$.
Thus, a solution with utilities $(0.9, 1.8, 2.7)$ is not considered approximately-optimal. Unfortunately, according to this definition, even the leximin-optimal solution --- with utilities $(1,2,3)$  --- is not considered approximately-optimal.
We expect a good approximation notion to be a relaxation of leximin-optimality.

\subsection{Our Definition}\label{sub:our-def}
 Let $\DEFmultApprox\in (0,1]$ and $\DEFadditiveApprox \geq 0$ be multiplicative and additive approximation factors, respectively.

\paragraph{Comparison of values}
As we focus on maximization problems, 
%
%
given two values $v_2 \geq v_1 \geq 0$, we say that $v_1$ approximates $v_2$ if $v_1 \geq \DEFmultApprox \cdot v_2 - \DEFadditiveApprox$.
In this case, $v_1$ is an approximate replacement for $v_2$.
However, when $v_1 < \DEFmultApprox \cdot v_2 - \DEFadditiveApprox$, we say that $v_2$ is \emph{$(\DEFmultApprox,\DEFadditiveApprox)$-substantially-higher} than $v_1$.
In this case, $v_1$ is smaller than any $(\DEFmultApprox,\DEFadditiveApprox)$-approximation of $v_2$. 

\paragraph{The approximate leximin order} 
The first step is defining the following \textit{partial} order\footnote{A proof that the approximate leximin order is a strict partial order can be found in Appendix \appendixName{\ref{sec:approx-order-is-strict-partial}}{A}\fullVer.}:
a solution $y$ is \emph{$(\DEFmultApprox,\DEFadditiveApprox)$-leximin-preferred} over a solution $x$, denoted $y \alphaBetaPreferred x$, if there exists an integer $1 \leq k \leq n$ such that the smallest $(k-1)$ objective values of $y$ are \emph{at least} those of $x$, and the $k$'th smallest objective value of $y$ is $(\DEFmultApprox,\DEFadditiveApprox)$-substantially-higher than the $k$'th smallest objective value of $x$, that is:
\begin{align*}
    \forall j < k \colon \quad &\valBy{j}{y} \geq \valBy{j}{x}\\
    &\valBy{k}{y} > \frac{1}{\DEFmultApprox} \left( \valBy{k}{x} + \DEFadditiveApprox\right)
\end{align*}
A maximal element of this order is a solution over which no solution is $(\DEFmultApprox,\DEFadditiveApprox)$-leximin-preferred.
For clarity, we define the corresponding relation set as follows:
\begin{align*}
    \relationSetAlphaBeta = \{(y,x) \mid  x,y \in S, \Hquad y \alphaBetaPreferred x\}
\end{align*}

Before describing the approximation definition, we present two observations about this relation that will be useful later, followed by an example to illustrate how it works.
The proofs are straightforward and are omitted. 

The first observation is that the leximin order is equivalent to the approximate leximin order for $\DEFmultApprox = 1$ and $\DEFadditiveApprox = 0$ (that is, in the absence of errors).

\begin{lemma}\label{lemma:approx-relation-prop1}
    Let $x,y \in S$. Then, $y \leximinPreferred x \iff y \alphaBetaPreferredParams{1}{0} x$
\end{lemma}

The second observation relates different approximate leximin orders according to their \emph{error} factors.
Notice that, for additive errors, $\DEFadditiveError$ also describes the error size; whereas for multiplicative errors, one minus $\DEFmultApprox$ describes it.
Throughout the remainder of this section, we denote the multiplicative \emph{error factor} by $\DEFmultError = 1-\DEFmultApprox$.

\begin{observation}\label{obs:approx-relation-prop2}
     Let $0 \leq \DEFmultErrorOf{\DEFmultApprox_1} \leq  \DEFmultErrorOf{\DEFmultApprox_2} < 1$ and $0 \leq \DEFadditiveError_1 \leq \DEFadditiveError_2$. 
     Then, $y \alphaBetaPreferredParams{\DEFmultApprox_2}{\DEFadditiveApprox_2} x \Rightarrow y \alphaBetaPreferredParams{\DEFmultApprox_1}{\DEFadditiveApprox_1} x$.
\end{observation}
One can easily verify that it follows directly from the definition as
$\frac{1}{\DEFmultApprox_2} \geq \frac{1}{\DEFmultApprox_1}$. 
Accordingly, by considering the relation sets $\relationSetParams{\DEFmultApprox_1}{\DEFadditiveApprox_1}$ and $\relationSetParams{\DEFmultApprox_2}{\DEFadditiveApprox_2}$, we can conclude that $\relationSetParams{\DEFmultApprox_2}{\DEFadditiveApprox_2} \subseteq \relationSetParams{\DEFmultApprox_1}{\DEFadditiveApprox_1}$.
This means that as the \emph{error} parameters $\DEFmultError$ and $\DEFadditiveApprox$ increase,
the relation becomes \emph{more partial}:
when $\DEFmultError = 0$ and $\DEFadditiveApprox = 0$ it is a total order, any two elements that yield different utilities appear as a pair in $\relationSetParams{1}{0}$; but as they increase, the set $\relationSetAlphaBeta$ potentially becomes smaller, as fewer pairs are comparable.

\paragraph{Example } To illustrate, consider a group of $3$ agents, that has to select one out of three options $x,y,z$, with sorted utility vectors $\sortedValues{x}=(1,10,15), \sortedValues{y} =(1,40,60), \sortedValues{z}=(2,20,30)$.
Table \ref{table:relationSets} indicates what is $\relationSetParams{\DEFmultApprox}{\DEFadditiveApprox}$ for different choices of $\DEFmultApprox$ and $\DEFadditiveApprox$.
\displayEcai{

}
It is easy to verify that, indeed, $\relationSetParams{1}{0}$ is a total order --- $(z,x), (z,y) \in \relationSetParams{1}{0}$ since $2>1$ and $(y, x) \in \relationSetParams{1}{0}$ since $1=1$ and $40>10$.
\displayEcai{

}
In accordance with Observation \ref{obs:approx-relation-prop2}, the relation set remains the same or becomes smaller as either $\DEFmultApprox$ decreases (and the error factor $\DEFmultError$ increases) or $\DEFadditiveError$ increases.
As an example, we provide a partial calculation of $\relationSetParams{0.75}{1}$.
\displayEcai{

}
First, by Observation \ref{obs:approx-relation-prop2}, we know that $\relationSetParams{0.75}{1} \subseteq \relationSetParams{1}{0}$, and so, it is sufficient to consider only the pairs in $\relationSetParams{1}{0}$. 
\displayEcai{

}
Consider the pair $(z,x)$.
In order to be included in the relation set, there must be a $1 \leq k \leq 3$ that meets the requirements.
For $k=1$, as $2 \ngtr \frac{1}{0.75}(1+1)$, the requirement for $k$ does not hold.
However, for $k=2$, it does. As $2 \geq 1$, the requirement for $i<k$ holds; and as $20 > \frac{1}{0.75}(10+1)$, the requirement for $k$ holds.
Therefore, $(z,x) \in \relationSetParams{0.75}{1}$.
Similarly, one can check that
$(y,x) \in \relationSetParams{0.75}{1}$.
\displayEcai{

}
Next, consider the pair $(z,y)$.
For $k=1$, as before, since $2 \ngtr \frac{1}{0.75}(1+1)$, the requirement for $k$ does not hold.
For $k =2$ and $k=3$, it is sufficient to notice that $20 < 40$, therefore the requirements for both does not hold. And so, $(z,y) \notin \relationSetParams{0.75}{1}$.

\begin{table}
\centering
{\caption{Different relation sets result from different choices of $\DEFmultApprox$ and $\DEFadditiveApprox$ in the example above.
Each cell contains the corresponding relation set $\relationSetParams{\DEFmultApprox}{\DEFadditiveApprox}$.}
\label{table:relationSets}}
\begin{tabular}{|c |c c c c|}
 \hline
\backslashbox{$\DEFmultApprox$}{$\DEFadditiveApprox$}& 0 & 1 & 15 & 45 \\ 
 \hline
 1 & \{(z,x),(z,y),(y,x)\} & \{(z,x),(y,x)\} & \{(y,x)\} & \{\}\\ 
 0.75 & \{(z,x),(z,y),(y,x)\} & \{(z,x),(y,x)\} & \{(y,x)\} &\{\}\\
 0.5 & \{(y,x)\} & \{(y,x)\} & \{\}&\{\} \\
 0.25 & \{\} & \{\} & \{\}&\{\} \\ [1ex]
 \hline
\end{tabular}
\end{table}



The leximin approximation can now be defined.

\paragraph{Leximin approximation}
We say that a solution $x\in S$ is \emph{$(\DEFmultApprox,\DEFadditiveError)$-approximately leximin-optimal} if it is a maximum element of the order $\alphaBetaPreferred$\displayComsoc{.} \displayEcai{ in $S$ for $\DEFmultApprox\in (0,1]$ and $\DEFadditiveError \geq 0$.
That is, if there is \emph{no} solution in $S$ that is $(\DEFmultApprox,\DEFadditiveError)$-leximin-preferred over it.}
For brevity, we use the term \emph{leximin approximation} to describe an approximately leximin-optimal solution.

This definition has some important properties.
Lemma \ref{lemma:absence-of-errors} proves that in the absence of errors ($\DEFmultError = \DEFadditiveError = 0$) it is equivalent to the exact leximin optimal definition. 
Then, Lemma \ref{lemma:beta1-beta2-approx} shows that an $(\DEFmultApprox_1,\DEFadditiveError_1)$-leximin-approximation is also an $(\DEFmultApprox_2,\DEFadditiveError_2)$-leximin-approximation when $0 \leq \DEFmultErrorOf{\DEFmultApprox_1} \leq  \DEFmultErrorOf{\DEFmultApprox_2} < 1$ and $0 \leq \DEFadditiveError_1 \leq \DEFadditiveError_2$.
Finally, Lemma \ref{lemma:exact-is-always-optimal} proves that a leximin optimal solution is also a leximin approximation for all factors.

\begin{lemma}\label{lemma:absence-of-errors}
 A solution is a $(1,0)$-leximin-approximation if and only if it is leximin optimal.
\end{lemma}
\begin{proof}
    %
    By definition, a solution $x^*$ is a $(1,0)$-leximin-approximation if and only if $x \nAlphaBetaPreferredParams{1}{0} x^*$ for any solution $x \in S$.
    This holds if and only if $x \nLeximinPreferred x^*$ for any solution $x \in S$ (by Lemma \ref{lemma:approx-relation-prop1}).
    Thus, by definition, $x^*$ is also leximin optimal.
\end{proof}

\begin{lemma}\label{lemma:beta1-beta2-approx}
    Let $0 \leq \DEFmultErrorOf{\DEFmultApprox_1} \leq  \DEFmultErrorOf{\DEFmultApprox_2} < 1$, $0 \leq \DEFadditiveError_1 \leq \DEFadditiveError_2$, and $x \in S$ be an $(\DEFmultApprox_1,\DEFadditiveError_1)$-leximin-approximation. Then $x$ is also an $(\DEFmultApprox_2,\DEFadditiveError_2)$-leximin-approximation.
\end{lemma}

\begin{proof}
    Since $x$ is an $(\DEFmultApprox_1,\DEFadditiveError_1)$-leximin-approximation, by definition, $y \nAlphaBetaPreferredParams{\DEFmultApprox_1}{\DEFadditiveError_1} x$ for any solution $y \in S$.
    Observation \ref{obs:approx-relation-prop2} implies
        that
        $y \nAlphaBetaPreferredParams{\DEFmultApprox_2}{\DEFadditiveError_2} x$ 
    for any solution $y \in S$. This means, by definition, that $x$ is an $(\DEFmultApprox_2,\DEFadditiveError_2)$-leximin-approximation.
\end{proof}

\begin{lemma}\label{lemma:exact-is-always-optimal}
    Let $x^* \in S$ be a leximin optimal solution. Then $x^*$ is also an $(\DEFmultApprox,\DEFadditiveError)$-leximin-approximation for any $\DEFmultError \in [0,1)$  and $\DEFadditiveError \geq 0$.
\end{lemma}

\begin{proof}
    By Lemma \ref{lemma:absence-of-errors}, $x^*$ is an $(1,0)$-leximin-approximation.
    Thus, according to Lemma \ref{lemma:beta1-beta2-approx}, $x^*$ is also an $(\DEFmultApprox,\DEFadditiveError)$-leximin-approximation for any $0 \leq \DEFmultErrorOf{\DEFmultApprox} < 1$ and $\DEFadditiveError \geq 0$.
    %
\end{proof}

Using the example given previously, we shall now demonstrate that as the error parameters $\DEFmultError$ and $\DEFadditiveError$ increase, the quality of the approximation decreases.
Consider table \ref{table:relationSets} once again.

If the corresponding relation set for $\DEFmultApprox$ and $\DEFadditiveApprox$ is the total order $\{(z,x),(z,y),(y,x)\}$, the only solution over which no other solution is $(\DEFmultApprox,\DEFadditiveApprox)$-leximin-preferred is $z$. Therefore, $z$ is the only $(\DEFmultApprox,\DEFadditiveApprox)$-leximin-approximation for these factors.
Indeed, it is the only group decision that maximizes the welfare of the agent with the smallest utility.

If the corresponding relation set is either $\{(z,x),(y,x)\}$ or $\{(y,x)\}$, as no solution is $(\DEFmultApprox,\DEFadditiveApprox)$-leximin-preferred over $z$ and $y$, both are $(\DEFmultApprox,\DEFadditiveApprox)$-leximin-approximations.
For example, for $\DEFmultApprox=0.5$ and $\DEFadditiveApprox=0$, $z$ still maximizes the utility of the poorest agent (2), and $y$ gives the poorest agent a utility of $1$, which is acceptable as it is half the maximum possible value $(2)$, and subject to giving the poorest agent at least $1$, maximizes the second-smallest utility ($40$). In contrast, while $x$, too, gives the poorest agent utility $1$, its second-smallest utility is $10$, which is less than half the maximum possible in this case ($40$), and therefore, $x$ is not a $(\DEFmultApprox,\DEFadditiveApprox)$-leximin-approximation.

Lastly, if  the relation set is the empty set, then no solution is $(\DEFmultApprox,\DEFadditiveApprox)$-leximin-preferred over the other, and all are $(\DEFmultApprox,\DEFadditiveApprox)$-leximin-approximations.

\paragraph{Egalitarian generalization} 
Another property of our definition is that a leximin-approximation also approximates the optimal egalitarian welfare to the same approximation factors. Formally: 

\begin{observation}\label{obs:rppeox-lex-generalizes-approx-egal}
    Let $x$ be an $(\DEFmultApprox, \DEFadditiveApprox)$-leximin-approximation. Then, the egalitarian value of $x$ (i.e., $\min_{i \in [n]} f_i(x) = \valBy{1}{x}$) is an $(\DEFmultApprox, \DEFadditiveApprox)$-approximation of the optimal egalitarian value (i.e., $\max_{y \in S}\min_{i \in [n]} f_i(y)$).
\end{observation}

\section{Approximation Algorithm}\label{sec:algo-short}
We now present an algorithm for computing a leximin approximation. 
The algorithm is an adaptation of one of the algorithms of Ogryczak
and {\'{S}}liwi{\'{n}}ski~\cite{Ogryczak_2006} for finding exact leximin optimal solutions. 
%

\subsection{Preliminary: exact leximin-optimal solution}
Following the definition of leximin, the core algorithm for finding a leximin optimal solution is iterative, wherein one first maximizes the least objective function, then the second, and so forth. 
In each iteration, 
$t=1,\ldots,n$, 
it looks for the value that maximizes the $t$-th smallest objective, $z_t$, given that for any $i < t$ the $i$-th smallest objective is at least $z_i$ (the value that was computed in the $i$-th iteration).
The core, single-objective optimization problem is thus:
\begin{align}
 \max &&&\ztVar{x}   \tag{\progBasic}\label{eq:basic-OP}\\
        s.t. &&& (\text{\progBasic.1}) \Hquad x \in S \nonumber\\
              &&& (\text{\progBasic.2}) \Hquad \valBy{\ell}{x}\geq z_{\ell} & \forXinY{\ell}{t-1} \nonumber \\
               &&& (\text{\progBasic.3}) \Hquad \valBy{t}{x} \geq \ztVar{x} \nonumber   
\end{align} 
where the variables are the scalar $\ztVar{x}$ and the vector $x$, whereas $z_1, \ldots z_{t-1}$ are constants (computed in previous iterations).

Suppose we are given a procedure $\textsf{OP}(z_1,\ldots,z_{t-1})$, which, given  $z_1,\ldots,z_{t-1},$ outputs $(x,z_t)$ that is the exact optimal solution to \eqref{eq:basic-OP}.  
Then, the \emph{leximin} optimal solution is obtained by iterating this process for $t=1,\ldots,n$, as described in Algorithm \ref{alg:basic-ordered-Outcomes}.
\displayEcai{
The algorithm first maximizes the smallest objective $\valBy{1}{x}$, and puts the result in $z_1$.
Then 
it maximizes the second-smallest objective $\valBy{2}{x}$,
subject to $\valBy{1}{x}$ being at least $z_1$, and puts the result in $z_2$; and so on.
}
\begin{algorithm}[!tbp]
\caption{The Ordered Outcomes Algorithm}
\label{alg:basic-ordered-Outcomes}
\begin{algorithmic}[1] 
\FOR{$t=1$ to $n$}
\STATE \( 
(x_t,z_t)\leftarrow \textsf{OP}(z_1,\ldots,z_{t-1})
\) 
\ENDFOR
\STATE \textbf{return} {$x_n$ (with objective values $f_1(x_n),\ldots,f_n(x_n)$)}.
\end{algorithmic}
\end{algorithm}

Since constraints (\progBasic.2) and (\progBasic.3) are not linear  with respect to the objective-functions, it is difficult to solve the program \eqref{eq:basic-OP} as is. 
Thus, \cite{Ogryczak_2006} suggests a way to ``linearize`` the program in two steps. 
%
%
\displayEcai{

}
First,
we replace \eqref{eq:basic-OP} with the following program, that considers sums instead of individual values (where again the variables are $\ztVar{x}$ and $x$):
\begin{align*}
\max &&&\ztVar{x} \tag{\progSums}\label{eq:sums-OP}\\
s.t. &&& (\text{\progSums.1}) \quad x \in S \\
&&& (\text{\progSums.2}) \quad \sum_{i \in F'} f_i(x) \geq \sum_{i=1}^{|F'|}  z_i && \forall F' \subseteq [n], \Hquad |F'| < t \\
&&& (\text{\progSums.3}) \quad \sum_{i \in F'} f_i(x) \geq \sum_{i=1}^{t-1}  z_i + z_t  && \forall F' \subseteq [n], \Hquad |F'| = t
\end{align*}
Here, constraints (\progBasic.2) and (\progBasic.3) are replaced with constraints (\progSums.2) and (\progSums.3), respectively. 
Constraint (\progSums.2) says that for any $\ell<t$, the sum of any $\ell$ objectives is at least the sum of the first $\ell$ constants $z_i$
(equivalently: the sum of the smallest $\ell$ objectives is at least the sum of the first $\ell$ constants $z_i$\footnote{A formal proof of this claim is given in Appendix \appendixName{\ref{sub:equivalence-P2-P3}}{B.2}\fullVer.}). 
Similarly, (\progSums.3) says that the sum of any $t$ objectives (equivalently: the sum of the smallest $t$ objectives) is at least the sum of the first $t-1$ constants  $z_i$, plus the variable $z_t$.

\displayEcai{
Suppose \eqref{eq:basic-OP}
is replaced with \eqref{eq:sums-OP}
in Algorithm\ref{alg:basic-ordered-Outcomes}.
Then, in the first iteration, the algorithm still maximizes the smallest objective $\valBy{1}{x}$, and puts the result in $z_1$.
In the second iteration, it maximizes the \emph{difference} between the sum of the two smallest objectives $\valBy{1}{x}+\valBy{2}{x}$ and $z_1$,
subject to $\valBy{1}{x}$ being at least $z_1$, and puts the result in $z_2$.
Since $z_1$ is the maximum value of $\valBy{1}{x}$, being at least $z_1$ becomes being exactly $z_1$, which means that, as 
 was for \eqref{eq:basic-OP}, the algorithm actually maximizes $\valBy{2}{x}$, subject to $\valBy{1}{x}$ being at least $z_1$.
%
%
%
%
Similarly for any iteration $1 <t\leq n$, as the sum $\sum_{i=1}^{\ell} z_i$ is the maximum value of $\sum_{i=1}^{\ell} \valBy{i}{x}$ for all $1\leq \ell < t$, it can be concluded that the algorithm actually maximizes $\valBy{t}{x}$, subject to $\valBy{\ell}{x}$ being at least $z_{\ell}$ for $1\leq \ell < t$ (as if \eqref{eq:basic-OP} was used).
Accordingly, the algorithm still finds a leximin-optimal solution.
}



While \eqref{eq:sums-OP} is linear with respect to the objective-functions, it has an exponential number of constraints.
To overcome this challenge, auxiliary variables were used in the second program ($y_{\ell}$ and $m_{\ell,j}$ for all $1 \leq \ell \leq t$ and $1 \leq  j\leq n$):
\begin{align}
\max &&& \ztVar{x} \tag{\progLinear}\label{eq:vsums-OP} \\
s.t. &&& (\text{\progLinear.1}) \Hquad x \in S \nonumber  \\
&&& (\text{\progLinear.2}) \Hquad \ell y_{\ell} - \sum_{j=1}^n m_{\ell,j}\geq \sum_{i=1}^{\ell}  z_i && \forXinY{\ell}{t-1} \nonumber \\
&&& (\text{\progLinear.3}) \Hquad t y_t - \sum_{j=1}^{n} m_{t,j} \geq \sum_{i=1}^{t-1}  z_i + z_t \nonumber \\
&&& (\text{\progLinear.4}) \Hquad m_{\ell,j} \geq y_{\ell} - f_j(x)  && \forXinY{\ell}{t},\Hquad \forXinY{j}{n} \nonumber \\
&&& (\text{\progLinear.5}) \Hquad m_{\ell,j} \geq 0  && \forXinY{\ell}{t},\Hquad \forXinY{j}{n} \nonumber
\end{align}
The meaning of the auxiliary variables in  \eqref{eq:vsums-OP} is explained in the proof of Lemma \ref{lem:equivalence} below.

The importance of the programs \eqref{eq:sums-OP} and \eqref{eq:vsums-OP} for leximin is shown by the following theorem (that combines  Theorem 4 in~\cite{Ogryczak2004TelecommunicationsND} and Theorem 1 in~\cite{Ogryczak_2006}):\begin{theorem*}
If Algorithm \ref{alg:basic-ordered-Outcomes} is applied with a solver for \eqref{eq:sums-OP} or \eqref{eq:vsums-OP}\footnote{If the algorithm uses a solver for \eqref{eq:vsums-OP}, it takes only the assignment of the variables $x$ and $z_t$
, ignoring the auxiliary variables.}
(instead of for \eqref{eq:basic-OP}), the algorithm still outputs a leximin-optimal solution. 
\end{theorem*}

We shall later see that our main result (Theorem \ref{th:main}) extends and implies their theorem.

\subsection{Using an approximate solver}
Now we assume that, instead of an exact solver in Algorithm \ref{alg:basic-ordered-Outcomes}, we only have an approximate solver. 
In this case, the constants $z_1,\ldots,z_{t-1}$ are only approximately-optimal solutions for the previous iterations.
It is easy to see that if \textsf{OP} is an $(\multApprox,\additiveApprox)$-approximation algorithm
to \eqref{eq:basic-OP}, then Algorithm \ref{alg:basic-ordered-Outcomes} outputs an $(\multApprox,\additiveApprox)$-leximin-approximation\footnote{A formal proof is given in Appendix \appendixName{\ref{sub:approximate-P1}}{B.1}\fullVer.}.
%

In contrast, for \eqref{eq:sums-OP} and \eqref{eq:vsums-OP}, we shall see that  Algorithm \ref{alg:basic-ordered-Outcomes} may output a solution that is \emph{not} an $(\multApprox,\additiveApprox)$-leximin-approximation.
However, we will prove that it is not too far from that --- in this case, the output is always an $\left(\frac{\multApprox^2}{1-\multApprox + \multApprox^2}, \frac{\additiveApprox}{1-\multApprox +\multApprox^2}\right)$-leximin-approximation.

To demonstrate both claims more clearly, we start by proving that the programs \eqref{eq:sums-OP} and \eqref{eq:vsums-OP} are \emph{equivalent} in the following sense:
\begin{lemma}
\label{lem:equivalence}
    Let $1 \leq t \leq n$ and let $z_1, \ldots z_{t-1} \in \mathbb{R}$.
    Then, $(x, z_t)$ is feasible for \eqref{eq:sums-OP} if and only if there exist $y_{\ell}$ and $m_{\ell,j}$ for $1 \leq \ell \leq t$ and $1 \leq j \leq n$ such that $\left(x, z_t, (y_1, \ldots, y_t), (m_{1,1}, \ldots m_{t,n})\right)$ is feasible for \eqref{eq:vsums-OP}.
\end{lemma}
The proof  is provided  in Appendix \appendixName{\ref{sub:equivalence-P2-P3}}{B.2}\fullVer.
\displayEcai{

}
Since both  \eqref{eq:sums-OP} and \eqref{eq:vsums-OP} have the same objective function ($\max z_t$), the lemma implies that $(x,z_t)$ is an $(\alpha,\epsilon)$-approximate solution for \eqref{eq:sums-OP} if and only if $(x,z_t)$ is a part of an $(\alpha,\epsilon)$-approximate solution for \eqref{eq:vsums-OP}.
Thus, it is sufficient to prove the theorems for only one of the problems. 
We will prove them for \eqref{eq:sums-OP}.

\begin{theorem}
There exist $\multApprox\in (0,1]$, $\additiveApprox \geq 0$ and \textsf{OP} that is an $(\multApprox,\additiveApprox)$-approximation procedure to \eqref{eq:sums-OP}, such that if Algorithm \ref{alg:basic-ordered-Outcomes} is applied with this procedure, it might return a solution that is not 
an $(\multApprox,\additiveApprox)$-leximin-approximation.
\end{theorem}
\begin{proof}
Consider the following multi-objective optimization problem with $n=2$:
\begin{align*}
    \max \quad &\{f_1(x) := x_1, f_2(x) := x_2\} \\
    s.t. \quad  & (1.1) \Hquad x_1 \leq 100, \quad (1.2) \Hquad x_1 + x_2 \leq 200, \quad (1.3) \Hquad x \in \mathbb{R}^2_{+}
\end{align*}
In the corresponding \eqref{eq:sums-OP}, constraint (\progSums.1) will be replaced with constraints (1.1)-(1.3).
The following is a possible run of the algorithm with \textsf{OP} that is a $(0.9,0)$-approximate solver.
\displayComsoc{
In iteration $t=1$, condition (\progSums.2) is empty, and the optimal value of $z_1$ is $100$, so \textsf{OP} may output $z_1=0.9\cdot 100 = 90$.
In iteration $t=2$, given  $z_1=90$, 
condition (\progSums.2) says that each of $x_1$ and $x_2$ must be at least $90$;
the optimal value of $z_2$ under these constraints is $110$, so \textsf{OP} may output $z_2=99$, for example with  $x_1=x_2=94.5$.  
Since $n=2$, the algorithm ends 
and returns the solution $(94.5,94.5)$.
}
\displayEcai{
\begin{itemize}
    \item In iteration $t=1$, condition (\progSums.2) is empty, and the optimal value of $z_1$ is $100$, so \textsf{OP} may output $z_1=0.9\cdot 100 = 90$.  
    \item
    In iteration $t=2$, given  $z_1=90$, 
condition (\progSums.2) says that each of $x_1$ and $x_2$ must be at least $90$;
the optimal value of $z_2$ under these constraints is $110$, so \textsf{OP} may output $z_2=99$, for example with  $x_1=x_2=94.5$.  
\item Since $n=2$, the algorithm ends 
and returns the solution $(94.5,94.5)$.
\end{itemize}
}
But $(x_1,x_2) = (94.5,105.5)$ is also a feasible solution, and it is $(0.9,0)$-leximin-preferred
since $105.4>\frac{1}{0.9}\cdot 94.5 = 105$.
Hence, the returned solution is \emph{not} a $(0.9,0)$-leximin-approximation.
\end{proof}

Note that, while the above solution is not a $(0.9,0)$-leximin-approximation, it is for $\multApprox = 0.896$.
Our main theorem below shows that this is not a coincidence:
using an approximate solver to \eqref{eq:sums-OP} or \eqref{eq:vsums-OP} in Algorithm \ref{alg:basic-ordered-Outcomes} guarantees a non-trivial leximin approximation.

\begin{theorem}\label{th:main}
Let $\multApprox\in (0,1]$, $\additiveApprox \geq 0$, and \textsf{OP} be an $(\multApprox,\additiveApprox)$-approximation procedure to \eqref{eq:sums-OP} or \eqref{eq:vsums-OP}. Then Algorithm \ref{alg:basic-ordered-Outcomes} outputs an $\left(\frac{\multApprox^2}{1-\multApprox + \multApprox^2}, \frac{\additiveApprox}{1-\multApprox +\multApprox^2}\right)$-leximin-approximation.  
\end{theorem}
For the above example, it guarantees an $(\frac{81}{91},0)\approx (0.89,0)$-leximin-approximation.

A complete proof of Theorem \ref{th:main} is given in Appendix \appendixName{\ref{sub:th:main}}{B.3}\fullVer.
Here we provide a high level overview of the main steps.
\displayEcai{

}
%
%
%
First, we note that the value of the variable $z_t$ is completely determined by the variable $x$. 
This is because the program aims to maximize $z_t$ that appears only in constraint (\progSums.3), which is equivalent to $z_t\leq \sum_{i=1}^{t} \valBy{i}{x} - \sum_{i=1}^{t-1}  z_i$.
Thus, this constraint will always hold with equality.
%
\displayEcai{

}
Next, 
we show that the returned solution, $\retSol$, is feasible to all single-objective
problems that were solved during the algorithm run. 
This allows us to relate the objective values attained by $\retSol$ and the constants $z_i$ values.
Specifically, for each iteration $t=1,\ldots, n$, the solver used in iteration $t$ is $(\alpha,\epsilon)$-approximately-optimal.
Therefore, we know that the optimal value of this program is at most $\frac{1}{\alpha}(z_t+\epsilon)$, where $z_t$ is the approximation obtained to that problem.
It then follows that the objective value attained by $\retSol$ for this program 
is also is at most $\frac{1}{\alpha}(z_t+\epsilon)$.

%
\displayEcai{

}
We then assume for contradiction that $\retSol$ is \emph{not} a leximin approximation as claimed in the theorem.
This means, by definition, that there exits a solution $y \in S$ and an integer $1 \leq k \leq n$ such that $\valBy{i}{y} \geq \valBy{i}{\retSol}$ for any $i<k$, while $\valBy{k}{y}$ is $\left(\frac{\multApprox^2}{1-\multApprox + \multApprox^2}, \frac{\additiveApprox}{1-\multApprox +\multApprox^2}\right)$--substantially-higher\footnote{See Section \ref{sub:our-def} for formal definition.} than $\valBy{k}{\retSol}$.
Accordingly, we prove that $y$ is feasible to the program that was solved in the $k$-th iteration, and that the objective value it attains in this problem is strictly higher than the optimal value $z_k^*$, which is a contradiction.

%


\displayComsoc{Theorem \ref{th:main} implies that if \textsf{OP} has only a multiplicative error ($\additiveApprox = 0$), the returned solution will also have only a multiplicative error, and if \textsf{OP} has only an additive error ($\multApprox = 1$), the returned solution will also have only the same additive error $\additiveError$.}
\displayEcai{Theorem \ref{th:main} implies that if \textsf{OP} only has a multiplicative error ($\additiveApprox = 0$), the solution returned by Algorithm \ref{alg:basic-ordered-Outcomes} will only have a multiplicative error as well, and if \textsf{OP} only has an additive error ($\multApprox = 1$), the solution returned by Algorithm \ref{alg:basic-ordered-Outcomes} will have only the same additive error $\additiveError$.}

\eden{Again: to remove or to rewrite?
Note that there are many cases in which the required procedure (\textsf{OP}) can be implemented easily. For example, when the $S$ is convex and all $f_i$’s are concave (and polynomially computable), using convex optimization techniques.}

\subsection{Using a randomized solver}
Next, we assume that the solver is not only approximate but  also \emph{randomized} --- it always 
returns a feasible solution,
but only 
with probability $p \in [0,1]$ the solution is also approximately-optimal.
As Algorithm \ref{alg:basic-ordered-Outcomes} activates the solver $n$ times overall, assuming the success events of different activations are independent, there is a probability of $p^n$ that
the solver returns an approximately-optimal solution in every iteration and so, Algorithm \ref{alg:basic-ordered-Outcomes} performs as in the previous subsection.
 This leads to the following conclusion:
\begin{corollary}\label{corollary:main-with-probability}
Let $\multApprox\in (0,1]$, $\additiveError \geq 0$, $p \in (0,1]$, and \textsf{OP} be a \emph{$p$-randomized} $(\multApprox, \additiveError)$-approximation procedure to \eqref{eq:sums-OP} or \eqref{eq:vsums-OP}. Then Algorithm \ref{alg:basic-ordered-Outcomes} outputs an $\left(\frac{\multApprox^2}{1-\multApprox + \multApprox^2}, \frac{\additiveApprox}{1-\multApprox +\multApprox^2}\right)$-leximin-approximation with probability $p^n$.
\end{corollary}

Notice that, since the procedure \textsf{OP} always returns a feasible solution to the single-objective problem, Algorithm \ref{alg:basic-ordered-Outcomes} always returns a feasible solution as well.

Further, it is easy to see that one can increase the success probability of Algorithm \ref{alg:basic-ordered-Outcomes} if, in each iteration, \textsf{OP} will be operated more times, and the solution with the highest objective will be used.

\section{Conclusion and Future Work}
\label{sec:future}
We presented a practical solution to the problem of leximin optimization when only an approximate single-objective solver is available. 
The algorithm is guaranteed to terminate in polynomial time, and its approximation ratio degrades gracefully as a function of the approximation ratio of the single-objective solver.

Currently, our algorithm handles two main settings. First, when inaccuracies in the single-objective solver stem from numeric errors.
Second, when the problem is convex and satisfy several assumptions.
It may be interesting to study more settings in which the inaccuracies stem from computational hardness of the single-objective problem.
%
%

In particular, to approximate the egalitarian welfare, it is common to model the problem as an integer program or as an exponential sized linear program (e.g., \cite{bansal2006santa, kawase_max-min_2020}) and then approximate the program using different techniques.
Can these algorithms be generalized to consider the additional constraints described in Section \ref{sec:algo-short}? This will allow approximating leximin using the approach in this paper.


Another question is whether it is possible to obtain a better approximation factor for leximin, given an $(\multApprox, \additiveApprox)$-approximation algorithm for the single-objective problem.
Specifically, can an $(\multApprox, \additiveApprox)$-approximation to leximin can be obtained in polynomial time? 
If not, what would be the best possible approximation in this case?

\displayComsoc{\section{Acknowledgments}}
\displayEcai{\newpage\ack} 
This research is partly supported by the Israel Science Foundation grants 712/20 and 2697/22.
We are grateful to Sylvain Bouveret for suggesting several alternative definitions and helpful insights. We sincerely appreciate Arkadi Nemirovski, Nikhil Bansal, Shaddin Dughmi, and Tom Tirer for providing helpful answers and clarifications. 
We are also grateful to the following members of the stack exchange network for their very helpful answers to our technical questions:
Neal Young,%
\footnote{
	\url{https://cstheory.stackexchange.com/questions/51206} and 
	\url{https://cstheory.stackexchange.com/questions/51003}
	and
	\url{https://cstheory.stackexchange.com/questions/52353}
}
1Rock,%
\footnote{
\url{https://math.stackexchange.com/questions/4466551}
}
Mark L. Stone%
\footnote{
\url{https://or.stackexchange.com/questions/8633}
and 
\url{https://or.stackexchange.com/questions/11007}
}
and Rob Pratt.%
\footnote{\url{https://or.stackexchange.com/questions/8980}}
Lastly, we would like to thank the reviewers in COMSOC 2023 and ECAI 2023 for their helpful comments.
\displayComsoc{\newpage}
\bibliography{references}

\clearpage
\appendix

\section{The Approximate Leximin Order}\label{sec:approx-order-is-strict-partial}

Unlike the leximin order, $\leximinPreferred$, which is a \textbf{total} order, the approximate leximin order, $\alphaBetaPreferred$ for $\DEFmultApprox\in (0,1]$ and $\DEFadditiveApprox \geq 0$ is a \textbf{partial} order.
The difference is that in partial orders, not all vectors are comparable.
Consider for example the sorted vectors $(1,2)$ and $(1, 3)$. 
According to the leximin order, $(1,3)$ is clearly preferred (as $3>2$), but according to many approximate leximin orders neither one is preferred over the other, for example according to the orders $\alphaBetaPreferredParams{0.6}{0}$,$ \alphaBetaPreferredParams{1}{1}$ or $\alphaBetaPreferredParams{0.8}{0.5}$.

An order is a strict partial order if it is irreflexive, transitive and asymmetric.
Lemma \ref{lemma:order-is-irreflexive} proves that the order is irreflexive, Lemma \ref{lemma:order-is-transitive} proves it is transitive, and Lemma \ref{lemma:order-is-asymmetric} proves that it is asymmetric.



Let $\DEFmultApprox\in (0,1]$ and $\DEFadditiveApprox \geq 0$. 

\begin{lemma}\label{lemma:order-is-irreflexive}
    The approximate leximin order $\alphaBetaPreferred$ is irreflexive.
\end{lemma}

\begin{proof}
    Let $x$ be a solution. We will show that $x \nAlphaBetaPreferred x$.
    As the definition requires that one component be \emph{strictly greater} than the other, it is trivial.
\end{proof}

\begin{lemma}\label{lemma:order-is-transitive}
    The approximate leximin order $\alphaBetaPreferred$ is transitive.
\end{lemma}

\begin{proof}
    Let $x,y$ and $z$ be solutions such that $x \alphaBetaPreferred y$ and $y \alphaBetaPreferred z$.
    We will prove that $x \alphaBetaPreferred z$.

    Since $x \alphaBetaPreferred y$, there exists an integer $ k_1 \in [n]$ such that:
    \begin{align*}
        \forall j<k_1 \colon &  \valBy{j}{x} \geq \valBy{j}{y}\\
            & \valBy{k_1}{x} > \frac{1}{\DEFmultApprox} \left( \valBy{k_1}{y} + \DEFadditiveApprox \right)
    \end{align*}
    And since $y \alphaBetaPreferred z$, there exists an integer $k_2 \in [n]$ such that:
    \begin{align*}
        \forall j<k_2 \colon &  \valBy{j}{y} \geq \valBy{j}{z}\\
            & \valBy{k_2}{y} > \frac{1}{\DEFmultApprox} \left( \valBy{k_2}{z} + \DEFadditiveApprox \right) 
    \end{align*}

    As $\DEFmultApprox \in (0,1]$ and $\DEFadditiveApprox \geq 0$, it follows that:
    \begin{align}\label{eq:trans-k-s}
        \valBy{k_1}{x} > \valBy{k_1}{y}, \Hquad \valBy{k_2}{y} >  \valBy{k_2}{z}
    \end{align}

    
    Let $k = \min\{k_1,k_2\}$.
    
    If $k = k_1$, by the definition of $k_1$, $\valBy{k}{x} > \frac{1}{\DEFmultApprox} \left( \valBy{k}{y} + \DEFadditiveApprox \right)$.
    However, $\valBy{k}{y} \geq \valBy{k}{z}$, by definition if $k<k_2$ and by Equation \eqref{eq:trans-k-s} if $k=k_2$. \ref{eq:transitive-k}
    Therefore, $\valBy{k}{x} > \frac{1}{\DEFmultApprox} \left( \valBy{k}{z} + \DEFadditiveApprox \right)$.
    
    Otherwise, if $k=k_2$, by the definition of $k_2$, $\valBy{k}{y} > \frac{1}{\DEFmultApprox} \left( \valBy{k}{z} + \DEFadditiveApprox \right)$. But, $\valBy{k}{x} \geq \valBy{k}{y}$, by definition if $k<k_1$ and by Equation \eqref{eq:trans-k-s} if $k=k_1$. Again, we can conclude that $\valBy{k}{x} > \frac{1}{\DEFmultApprox} \left( \valBy{k}{z} + \DEFadditiveApprox \right)$.

     In addition, for each $j<k$, since $j< k_1$ and $j < k_2$, by definition the following holds:
    \begin{align}\label{eq:transitive-k}
        \valBy{j}{x} \geq \valBy{j}{y} \geq \valBy{j}{z}
    \end{align}
    So, $k$ is an integer that satisfy all the requirements, and so, $x \alphaBetaPreferred z$.
    \end{proof}

    \begin{lemma}\label{lemma:order-is-asymmetric}
        The approximate leximin order $\alphaBetaPreferred$ is asymmetric.
    \end{lemma}
    
    \begin{proof}
        Let $x$ and $y$ be solutions such that $x \alphaBetaPreferred y$. We will show that $y \nAlphaBetaPreferred x$. 
        Assume by contradiction that $y \alphaBetaPreferred x$. 
        From Lemma \ref{lemma:order-is-transitive}, this relation is transitive. Therefore, since $x \alphaBetaPreferred y$ and $y \alphaBetaPreferred x$, also $x \alphaBetaPreferred x$.
        But, from Lemma \ref{lemma:order-is-irreflexive}, this relation is irreflexive --- a contradiction.
    \end{proof}
\section{Proofs Omitted From Section \ref{sec:algo-short}}\label{sec:algo-sec-proofs}

\subsection{Using an approximate solver for \eqref{eq:basic-OP}}
\label{sub:approximate-P1}
Recall that \eqref{eq:basic-OP} is described as follows:
\begin{align}
 \max &&&\ztVar{x}   \tag{\progBasic}\\
        s.t. &&& (\text{\progBasic.1}) \Hquad x \in S \nonumber\\
              &&& (\text{\progBasic.2}) \Hquad \valBy{\ell}{x}\geq z_{\ell} & \forall \ell \in [t-1] \nonumber \\
               &&& (\text{\progBasic.3}) \Hquad \valBy{t}{x} \geq \ztVar{x} \nonumber   
\end{align} 
This section proves the following lemma:
\begin{lemma}
    Let $\multApprox\in (0,1]$, $\additiveApprox \geq 0$, and \textsf{OP} be an $(\multApprox,\additiveApprox)$-approximation procedure to \eqref{eq:basic-OP}. Then Algorithm \ref{alg:basic-ordered-Outcomes} outputs an $\left(\multApprox, \additiveApprox\right)$-leximin-approximation.  
\end{lemma}

\begin{proof}
    Let $x^*$ be the returned solution and assume by contradiction that it is \emph{not} $\left(\multApprox, \additiveApprox\right)$-approximately leximin-optimal.
    This means that there exists a $y \in S$ that is $(\multApprox, \additiveApprox)$-leximin-preferred over it. 
    That is, there exists an integer $k \in [n]$ such that:
    \begin{align*}
        \forall i < k \colon \Hquad &\valBy{i}{y} \geq \valBy{i}{x^*}\\
        & \valBy{k}{y} > \frac{1}{\multApprox}\left(\valBy{k}{x^*}+\additiveApprox\right)
    \end{align*}
    
    Since $x^*$ is a solution to  \eqref{eq:basic-OP} that was solved in the iteration $t=n$, it must satisfy all its constraints, and therefore:
    \begin{align}\label{eq:p1-x-to-z}
        \forall i \in [n] \colon \Hquad \valBy{i}{x^*} \geq z_i
    \end{align}
    by constraint (\progBasic.2) for $i<n$ and by constraint (\progBasic.3) for $i=n$.
    
    As $y$'s smallest $k$ values are at least as those of $x^*$, we can conclude that for each $i\leq k$ the $i$-th smallest value of $y$ is at least $z_i$.
    Therefore, $y$ is feasible to  \eqref{eq:basic-OP} that was solved in the iteration $t=k$.
    
    During the algorithm run, $z_k$ was obtained as an $(\multApprox, \additiveApprox)$-approximation to \eqref{eq:basic-OP} that was solved in the iteration $t=k$ , and therefore, the optimal value of this problem
     is at most $\frac{1}{\multApprox}(z_k+\additiveApprox)$.
    But, the objective value $y$ yields in this problem is $\valBy{k}{y}$, which is higher than this value:
    \begin{align*}
        \valBy{k}{y} &> \frac{1}{\multApprox}\left(\valBy{k}{x^*}+\additiveApprox\right)\\
        &\geq \frac{1}{\multApprox}\left(z_k+\additiveApprox\right) && \text{(By Equation \eqref{eq:p1-x-to-z} for $i=k$)}
    \end{align*}
    This is a contradiction.
\end{proof}

\subsection{Equivalence of \eqref{eq:sums-OP} and \eqref{eq:vsums-OP}}
\label{sub:equivalence-P2-P3}
Recall the problems' descriptions:
 
\begin{align*}
\max &&&\ztVar{x} \tag{\progSums}\\
s.t. &&& (\text{\progSums.1}) \quad x \in S \\
&&& (\text{\progSums.2}) \quad \sum_{i \in F'} f_i(x) \geq \sum_{i=1}^{|F'|}  z_i && \forall F' \subseteq [n], \Hquad |F'| < t \\
&&& (\text{\progSums.3}) \quad \sum_{i \in F'} f_i(x) \geq \sum_{i=1}^{t-1}  z_i + z_t && \forall F' \subseteq [n], \Hquad |F'| = t
\end{align*}

\begin{align}
\max &&& \ztVar{x} \tag{\progLinear}\\
s.t. &&& (\text{\progLinear.1}) \Hquad x \in S \nonumber  \\
                    &&& (\text{\progLinear.2}) \Hquad \ell y_{\ell} - \sum_{j=1}^n m_{\ell,j}\geq \sum_{i=1}^{\ell}  z_i && \forXinY{\ell}{t-1} \nonumber \\
                    &&& (\text{\progLinear.3}) \Hquad t y_t - \sum_{j=1}^{n} m_{t,j} \geq \sum_{i=1}^{t-1}  z_i + z_t \nonumber \\
                    &&& (\text{\progLinear.4}) \Hquad m_{\ell,j} \geq y_{\ell} - f_j(x)  && \forXinY{\ell}{t},\Hquad \forXinY{j}{n} \nonumber \\
                    &&& (\text{\progLinear.5}) \Hquad m_{\ell,j} \geq 0  &&  \forXinY{\ell}{t},\Hquad \forXinY{j}{n} \nonumber
\end{align}

We use another equivalent representation of \eqref{eq:sums-OP}, which is more compact and will simplify the proofs, also introduced by \cite{Ogryczak_2006}:
\begin{align*}
    \max &&& z_t \tag{\progCompact}\label{eq:compact-OP}\\
    s.t. &&& (\text{\progCompact.1}) \Hquad x \in S\nonumber\\
                    &&& (\text{\progCompact.2}) \Hquad \sum_{i=1}^{\ell} \valBy{i}{x} \geq \sum_{i=1}^{\ell}  z_i && \forXinY{\ell}{t-1} \nonumber\\
                    &&& (\text{\progCompact.3}) \Hquad \sum_{i=1}^{t} \valBy{i}{x} \geq \sum_{i=1}^{t-1}  z_i + z_t
\end{align*}
In this problem, constraints (\progSums.2) and (\progSums.3) are replaced by (\progCompact.2) and (\progCompact.3), respectively.  
(\progSums.2) gives, for each $\ell$, a lower bound on the sum for \emph{any} set of $\ell$ objective functions; whereas (\progCompact.2) only considers the sum of the $\ell$ \emph{smallest} such values,
and similarly for (\progSums.3) and (\progCompact.3). 


This section proves that these \emph{three} problems are \emph{equivalent} in the following sense: 
\begin{lemma}\label{lem:equivalence-of-all-three}
    Let $t \in [n]$ and let $z_1, \ldots z_{t-1} \in \mathbb{R}$.
    Then, $(x, z_t)$ is feasible for \eqref{eq:sums-OP} if and only if $(x, z_t)$ is feasible for \eqref{eq:compact-OP} if and only if there exist $y_{\ell}$ and $m_{\ell,j}$ for $\ell \in [t]$ and $j \in [n]$ such that $\left(x, z_t, (y_1, \ldots, y_t), (m_{1,1}, \ldots m_{t,n})\right)$ is feasible for \eqref{eq:vsums-OP}.
\end{lemma} 
It is clear that this lemma implies Lemma \ref{lem:equivalence}, which only claims a part of it.

We start by proving that \eqref{eq:sums-OP} and \eqref{eq:compact-OP} are equivalent. That is, $(x, z_t)$ is feasible for \eqref{eq:sums-OP} if and only if $(x, z_t)$ is feasible for \eqref{eq:compact-OP}.
First, it is clear that $x$ satisfies constraint (\progSums.1) if and only if it satisfies constraint (\progCompact.1) (as both constraints are the same, $x \in S$).
To prove the other requirements, we start with the following lemma:
\begin{lemma}\label{lemma:sums-to-comp-constrants}
    For any $x \in S$, any $\ell \in [n]$ and a constant $c \in \mathbb{R}$ the following two conditions are equivalent:
    \begin{align}\label{eq:sums-to-comp-constrants}
         \forall F' \subseteq [n], |F'| = \ell \colon \sum_{i \in F'} f_i(x) &\geq c 
         \\
         \sum_{i=1}^{\ell} \valBy{i}{x}&\geq c 
    \end{align}
\end{lemma}

\begin{proof}
    For the first direction, recall that the values $ \valBy{1}{x}, \dots,  \valBy{\ell}{x}$ were obtained from $\ell$ objective functions (those who yield the smallest value).
    By the assumption, the sum of any set of function with size $\ell$ is at least $c$; therefore, it is true in particular for the functions corresponding to the values $ (\valBy{1}{x})_{i=1}^{\ell}$.
    For the second direction, assume that $\sum_{i=1}^{\ell} \valBy{i}{x}\geq c$.
    Since $ \valBy{1}{x}, \dots,  \valBy{\ell}{x}$ are the $\ell$ smallest objective values, we get that:
    \begin{align*}
       \forall F' \subseteq [n],\Hquad |F'| = \ell \colon \quad \sum_{i \in F'}f_i(x) \geq \sum_{i=1}^s \valBy{i}{x}\geq c.
    \end{align*}
\end{proof}

Accordingly, $x$ satisfies constraint (\progSums.2) --- for any $\ell \in [t-1]$, 
    \begin{align*}
        \forall F' \subseteq [n], |F'| = \ell \colon \sum_{i \in F'} f_i(x) \geq \sum_{i=1}^{\ell} z_i
    \end{align*}
    if and only if it satisfies $\sum_{i=1}^{\ell} \valBy{i}{x} \geq \sum_{i=1}^{\ell} z_i$, which is constraint (\progCompact.2).
    Similarly, $x$ and $z_t$ satisfy constraint (\progSums.3), 
    \begin{align*}
        \forall F' \subseteq [n], |F'| = t \colon \sum_{i \in F'} f_i(x) \geq \sum_{i=1}^{t} z_i
    \end{align*}
    if and only if $\sum_{i=1}^{t} \valBy{i}{x} \geq \sum_{i=1}^{t} z_i$, which is constraint (\progCompact.3).
    That is, $x$ ans $z_t$ satisfy all the constraints of \eqref{eq:sums-OP} if and only if they satisfy all the constraints of \eqref{eq:compact-OP}.

Now, we will prove that that \eqref{eq:compact-OP} and \eqref{eq:vsums-OP} are equivalent, that is, $(x, z_t)$ is feasible for \eqref{eq:compact-OP} if and only if there exist $y_{\ell}$ and $m_{\ell,j}$ for $\ell \in [t]$ and $j \in [n]$ such that $\left(x, z_t, (y_1, \ldots, y_t), (m_{1,1}, \ldots m_{t,n})\right)$ is feasible for \eqref{eq:vsums-OP}.

We start with the following lemma:
\begin{lemma}\label{lemma:comp-to-p3-m-sums}
    For any $x \in S$ and any constant $c \in \mathbb{R}$ (where $c$ does not depend on $j$),
    \begin{align*}
        \sum_{j=1}^n \max(0, c - f_j(x) ) = \sum_{j=1}^n \max(0, c - \valBy{j}{x} ).
    \end{align*}
\end{lemma}
\eden{TODO:
maybe to change to observation}
\begin{proof}
     Let $(\pi_1, \ldots, \pi_n)$ be a permutation of $\{1,\ldots,n\}$ such that $f_{\pi_i}(x) = \valBy{i}{x}$ for any $i \in [n]$ (notice that such permutation exists by the definition of $\valBy{}{}$).
     That is, the value that $f_{\pi_i}$ obtains is the $i$-th smallest one.
    Since each element in the sum $\sum_{j=1}^n \max(0, c - f_j(x))$ is affected by $j$ only through $f_j(x)$, the permutation $\pi$ allows us to conclude the following:
    \begin{align*}
        &\sum_{j=1}^n \max(0, c -f_j(x)) = \sum_{j=\pi_1}^{\pi_n} \max(0,c -f_j(x))\\ 
        &= \sum_{j=1}^n \max(0,c -f_{\pi_i}(x))  = \sum_{j=1}^{n} \max(0,c -\valBy{j}{x}).
    \end{align*}
\end{proof}

Lemma \ref{lemma:comp-to-p3-mapping} below proves the first direction of the equivalence between \eqref{eq:compact-OP} and \eqref{eq:vsums-OP}:
\begin{lemma}\label{lemma:comp-to-p3-mapping}
    Let $(x, z_t)$ be a feasible solution to \eqref{eq:compact-OP}.
    Then there exist $y_{\ell}$ and $m_{\ell,j}$ for $\ell \in [t]$ and $j \in [n]$ such that $\left(x, z_t, (y_1, \ldots, y_t), (m_{1,1}, \ldots m_{t,n})\right)$ is feasible for \eqref{eq:vsums-OP}.
\end{lemma}

\begin{proof}
    For any $\ell \in [t]$ and $j \in [n]$ define $y_{\ell}$ and $m_{\ell,j}$ as follows:
    \begin{align*}
        \quad y_{\ell} &:= \valBy{\ell}{x}
        \\
        m_{\ell,j} &:= \max(0,\valBy{\ell}{x} -f_j(x))
    \end{align*}
    
    First, since $x$ satisfies constraint (\progCompact.1), it is also satisfies constraint (\progLinear.1) of (as both constraints are the same and include only $x$).
    
    In addition, based on the choice of $y$ and $m$, it is clear that  $m_{\ell,j} \geq 0$ and $m_{\ell,j} \geq \valBy{\ell}{x} - f_j(x) = y_{\ell} - f_j(x)$ for any $\ell \in [n]$ and $j \in [n]$.
    Therefore, this assignment satisfies constraints (\progLinear.4) and (\progLinear.5).
    
    To show that this assignment also satisfies constraints (\progLinear.2) and (\progLinear.3), we first prove that for any $\ell \in [n]$ this assignment satisfies the following equation:
    \begin{align}\label{eq:comp-to-p3}
        \ell y_{\ell} - \sum_{j=1}^n m_{\ell,j} = \sum_{i=1}^{\ell} \valBy{i}{x} 
    \end{align}
    By the choice of $m$,  $\sum_{j=1}^n m_{\ell,j} = \sum_{j=1}^n \max(0, \valBy{\ell}{x} - f_j(x))$, and therefore, by Lemma \ref{lemma:comp-to-p3-m-sums}, it equals to $\sum_{j=1}^{n} \max(0,\valBy{\ell}{x} -\valBy{j}{x})$.
    \eden{\eden{TODO:
to try to add: as $\valBy{\ell}{x}$ does not depend on $j$}}
    Since $\valBy{\ell}{x}$ is the $\ell$-th smallest objective, it is clear that $\valBy{\ell}{x} - \valBy{j}{x} \leq 0$ for any $j > \ell$, and $\valBy{\ell}{x} - \valBy{j}{x} \geq 0$ for any $j \leq \ell$.
    And so, $\sum_{j=1}^n m_{\ell,j} = \ell\cdot \valBy{\ell}{x} - \sum_{i=1}^{\ell} \valBy{i}{x}$:
    \begin{align*}
        &\sum_{j=1}^n m_{\ell,j} = \sum_{j=1}^{n} \max(0,\valBy{\ell}{x} -\valBy{j}{x})\\
        &=\sum_{j=1}^{\ell} \max(0,\valBy{\ell}{x} -\valBy{j}{x}) + \sum_{j=\ell+1}^n \max(0,\valBy{\ell}{x} -\valBy{j}{x})\\
        &=\sum_{j=1}^{\ell} (\valBy{\ell}{x} -\valBy{j}{x}) + \sum_{j=\ell+1}^n 0 = \ell\cdot \valBy{\ell}{x} - \sum_{i=1}^{\ell} \valBy{i}{x}
    \end{align*}
    We can now conclude Equation \eqref{eq:comp-to-p3}:
    \begin{align*}
        \ell y_{\ell} - \sum_{j=1}^n m_{\ell,j} &= \ell \cdot \valBy{\ell}{x} - \ell\cdot \valBy{\ell}{x} + \sum_{i=1}^{\ell} \valBy{i}{x}\\
        &=\sum_{i=1}^{\ell} \valBy{i}{x}.
    \end{align*}

    Now, since $x$ satisfies constraint (\progCompact.2),  $\sum_{i=1}^{\ell} \valBy{i}{x} \geq \sum_{i=1}^{\ell} z_i$ for any $\ell \in [t-1]$.
    Therefore, by Equation \eqref{eq:comp-to-p3}, $\ell\cdot y_{\ell} - \sum_{j=1}^n m_{\ell,j}\geq  \sum_{i=1}^{\ell} z_i$ for any $\ell \in [t-1]$ and this assignment satisfies constraint (\progLinear.2).
    Similarly, as $x$ and $z_t$ satisfy constraint (\progCompact.3), $\sum_{i=1}^{t} \valBy{i}{x} \geq \sum_{i=1}^{t} z_i$, and so by Equation \eqref{eq:comp-to-p3}, $t \cdot y_{t} - \sum_{j=1}^n m_{t,j}\geq  \sum_{i=1}^{t} z_i$.
    This means that it also satisfies constraints (\progLinear.3).
\end{proof}

Finally, Lemma \ref{lemma:comp-to-p3-is-bij} below proves the second direction of the equivalence:

\begin{lemma}\label{lemma:comp-to-p3-is-bij}
    Let $\left(x, z_t, (y_1, \ldots, y_t), (m_{1,1}, \ldots m_{t,n})\right)$ be a feasible solution to \eqref{eq:vsums-OP}.
   Then, $(x, z_t)$ is feasible for \eqref{eq:compact-OP}.
\end{lemma}

\begin{proof}
    It is easy to see that since $x$ satisfies constraint (\progLinear.1), it is also satisfies constraint (\progCompact.1) (as both are the same).
    To show that it also satisfies constraints (\progCompact.2) and (\progCompact.3), we start by proving that for any $\ell \in [n]$:
    \begin{align}\label{eq:p3-to-comp}
          \sum_{i=1}^{\ell} \valBy{i}{x} \geq \ell y_{\ell} - \sum_{j=1}^n m_{\ell,j}
    \end{align}
    Suppose by contradiction that $ \sum_{i=1}^{\ell} \valBy{i}{x} < \ell y_{\ell} - \sum_{j=1}^n m_{\ell,j}$.

    For any $j\in [n]$ and any $\ell \in [n]$, $m_{\ell,j} \geq y_{\ell} - f_j(x)$ by constraint (\progLinear.4), and also $m_{\ell,j} \geq 0$ by constraint (\progLinear.5).
    Therefore, $m_{\ell,j} \geq \max(0,y_{\ell} -f_j(x))$.
    And so, by Lemma \ref{lemma:comp-to-p3-m-sums}, for any $\ell \in [t]$, the sum of $m_{\ell,j}$ over all $j \in [n]$ can be described as follows:
    \eden{TODO:
to try to add: as $\valBy{\ell}{x}$ does not depend on $j$}
    \begin{align}\label{eq:p3-to-conp-m-sum}
        \sum_{j=1}^n m_{\ell,j} \geq  \sum_{j=1}^n \max(0,y_{\ell} -f_j(x)) = \sum_{j=1}^n \max(0,y_{\ell} -\valBy{j}{x})
    \end{align}
    Therefore, $\sum_{i=1}^{\ell} \valBy{i}{x} <  \ell y_{\ell} - \sum_{j=1}^n \max(0,y_{\ell} -\valBy{j}{x})$ .
    Which means that:
\begin{align*}
    & \ell y_{\ell} - \sum_{i=1}^{\ell} \valBy{i}{x} - \sum_{j=1}^n \max(0,y_{\ell} -\valBy{j} {x}) > 0
\end{align*}
However, we will now see that the value of this expression is at most $0$, which is a contradiction:
\begin{align*}        
    &\ell y_{\ell} - \sum_{i=1}^{\ell} \valBy{i}{x} - \sum_{j=1}^n \max(0,y_{\ell} -\valBy{j} {x}) \\
    &= \sum_{i=1}^{\ell} y_{\ell} - \sum_{i=1}^{\ell} \valBy{i}{x} - \sum_{j=1}^n \max(0,y_{\ell} -\valBy{j} {x})\\
    & \equWithExp{\text{Since max with $0$ is at least $0$}}{\leq \sum_{i=1}^{\ell}( y_{\ell} - \valBy{i}{x} ) - \sum_{j=1}^{\ell} \max(0,y_{\ell} -\valBy{j} {x}) - \sum_{j=\ell+1}^n 0}\\
        &=  \sum_{j=1}^{\ell} \left((y_{\ell} - \valBy{j}{x}) - \max(0,y_{\ell} -\valBy{j}{x})\right)\\
        &\equWithExp{\text{Since each element\displayEcai{ of this sum} is at most $0$}}{\leq 0}
   \end{align*}
This is a contradiction; so \eqref{eq:p3-to-comp} is proved.

    Now, by constraint (\progLinear.2), $\ell y_{\ell} - \sum_{j=1}^n m_{\ell,j}\geq  \sum_{i=1}^{\ell} z_i$  for any $\ell \in [t-1]$. Therefore, by \eqref{eq:p3-to-comp}, also $\sum_{i=1}^{\ell} \valBy{i}{x} \geq \sum_{i=1}^{\ell} z_i$, which means that $x$ satisfies constraint (\progCompact.2).
    Similarly, by constraint (\progLinear.3),  $t y_{t} - \sum_{j=1}^n m_{t,j}\geq  \sum_{i=1}^{t} z_i$, and so by \eqref{eq:p3-to-comp}, also $\sum_{i=1}^{t} \valBy{i}{x} \geq \sum_{i=1}^{t} z_i$.
    This means that $x$ and $z_t$ satisfy constraint (\progCompact.3).
\end{proof}

\subsection{Proof of Theorem \ref{th:main}}
\label{sub:th:main}
This section is dedicated to proving Theorem \ref{th:main}:
let $\multApprox\in (0,1]$, $\additiveApprox \geq 0$, and \textsf{OP} be an $(\multApprox,\additiveApprox)$-approximation procedure to \eqref{eq:sums-OP} or \eqref{eq:vsums-OP}. Then Algorithm \ref{alg:basic-ordered-Outcomes} outputs an $\left(\frac{\multApprox^2}{1-\multApprox + \multApprox^2}, \frac{\additiveApprox}{1-\multApprox +\multApprox^2}\right)$-leximin-approximation. 

Based on Lemma \ref{lem:equivalence-of-all-three}, it is sufficient to prove the theorem for \eqref{eq:compact-OP}.

We start by observing that the value of the variable $z_t$ is completely determined by the variable $x$. This is because $z_t$ only appears in the last constraint, which is equivalent to 
$z_t \leq \sum_{i=1}^{t} \valBy{i}{x} - \sum_{i=1}^{t-1} z_i$. 
Therefore, for every $x$ that satisfies the first two constraints, 
it is possible to satisfy the last constraint by setting $z_t$ to any value which is at most
$\sum_{i=1}^{t} \valBy{i}{x} - \sum_{i=1}^{t-1} z_i$.
Moreover, as the program aims to maximize $z_t$, it will necessarily 
set $z_t$ to be equal to that expression, since $z_t$ is maximized when the constraint holds with equality. This is summarized in the observation below:

\begin{observation}
\label{obs:feasi-and-constraint2}
For any $t\geq 1$ and any constants $z_1,\ldots, z_{t-1}$,
every vector $x$ that satisfies constraints (\progCompact.1)  and (\progCompact.2)
is a part of a feasible solution $(x,z_t)$ for $z_t = \sum_{i=1}^{t} \valBy{i}{x} - \sum_{i=1}^{t-1} z_i$.
Moreover, 
\label{obs:obj-value}
the objective value  obtained by a feasible solution $x$ to the problem \eqref{eq:compact-OP} solved in iteration $t$ is $z_t = \sum_{i=1}^{t} \valBy{i}{x} - \sum_{i=1}^{t-1}  z_i$.
\end{observation}


Based on Observation \ref{obs:obj-value}, we can now slightly abuse the terminology and say that a solution $x$ is ``feasible`` in iteration $t$ if it satisfies constraints (\progCompact.1)  and (\progCompact.2) of the program solved in iteration $t$.

We denote $\retSol := x_n$ = the solution $x$ attained at the last iteration ($t=n$) of the algorithm. 
Since $\retSol$ is a feasible solution of \eqref{eq:compact-OP} in iteration $n$, and as each
iteration only adds new constraints to (\progCompact.2), it follows that $\retSol$ is also a feasible solution of \eqref{eq:compact-OP} in any iteration $1 \leq t\leq n$. 
\begin{observation}\label{obs:retSol-solves-any-t}
$\retSol$ is a feasible solution of \eqref{eq:compact-OP} in any iteration $1 \leq t\leq n$.
\end{observation}

Lastly, as the value obtained as $(\multApprox, \additiveApprox)$-approximation for this problem is the \emph{constant} $z_t$, the optimal value is at most $\frac{1}{\multApprox} (z_t+\additiveError)$. 
Consequently, the objective value of any feasible solution is at most this value.
Since $\retSol$ is feasible for any iteration $t$ (Observation \ref{obs:retSol-solves-any-t}) and since the objective value corresponding to $\retSol$ is $\sum_{i=1}^t \valBy{i}{\retSol} - \sum_{i=1}^{t-1} z_i$ (Observation \ref{obs:obj-value}), we can conclude:

\begin{observation}\label{obs:obj-xt-to-zt}
    The objective value obtained by $\retSol$ to the problem \eqref{eq:compact-OP} that was solved in iteration $t$ is at most $\frac{1}{\multApprox} (z_t+\additiveError)$. That is:
    \begin{align*}
        \sum_{i=1}^t \valBy{i}{\retSol} - \sum_{i=1}^{t-1} z_i \leq \frac{1}{\multApprox} \left(z_t+\additiveError \right).
    \end{align*}
\end{observation}


We start with Lemmas \ref{lemma:beta-vk}-\ref{lemma:fk-to-all}, which establish a relationship between the $k$-th least objective value obtained by $\retSol$ 
and the difference between the sum of the $(k-1)$ least objective values obtained by $\retSol$ and the sum of the $(k-1)$ first $z_i$ values.
Theorem \ref{th:main} then uses this relation to prove that the existence of another solution that would be $\left(\frac{\multApprox^2}{1-\multApprox + \multApprox^2}, \frac{\additiveApprox}{1-\multApprox +\multApprox^2}\right)$-leximin-preferred over $\retSol$ would lead to a contradiction.

For clarity, throughout the proofs, we denote the multiplicative \emph{error} factor by $\multError = 1-\multApprox$.


\begin{lemma}\label{lemma:beta-vk}
    For any $1 \leq k\leq n$, 
    \begin{align*}
        &\multError \sum_{i=1}^{k} \valBy{i}{\retSol} - \multError \sum_{i=1}^{k-1} z_i \geq \sum_{i=1}^k \valBy{i}{\retSol} - \sum_{i=1}^k z_i - \additiveError
    \end{align*}
\end{lemma}

\begin{proof}
By Observation \ref{obs:obj-xt-to-zt},
    \begin{align*}
         &\sum_{i=1}^k \valBy{i}{\retSol} - \sum_{i=1}^{k-1} z_i \leq \frac{1}{\multApprox} \left(z_k + \additiveError \right) = \frac{1}{1-\multError} \left(z_k + \additiveError \right)\\
         &\Rightarrow (1-\multError) \left(\sum_{i=1}^{k} \valBy{i}{\retSol} - \sum_{i=1}^{k-1}  z_i\right) \leq z_k +\additiveError\\
         &\Rightarrow \sum_{i=1}^{k} \valBy{i}{\retSol} - \sum_{i=1}^{k-1}  z_i - \multError \sum_{i=1}^{k} \valBy{i}{\retSol} + \multError \sum_{i=1}^{k-1}  z_i \leq z_k +\additiveError\\
         &\Rightarrow \sum_{i=1}^{k} \valBy{i}{\retSol} - \sum_{i=1}^{k}  z_i - \additiveError \leq \multError \sum_{i=1}^{k} \valBy{i}{\retSol} - \multError \sum_{i=1}^{k-1}  z_i. 
    \end{align*}
        \qedhere
\end{proof}

\begin{lemma}\label{lemma:beta-sums-to-diff}
    For any $1 \leq k \leq n$, 
    \begin{align*}
        \sum_{i=1}^k \multError^{i} \valBy{k-i+1}{\retSol} \geq \sum_{i=1}^k \valBy{i}{\retSol} - \sum_{i=1}^{k} z_i -\frac{1}{1 - \multError}\additiveError
    \end{align*}
\end{lemma}

\begin{proof}
    The proof is by induction on $k$.
    For $k=1$ the claim follows directly from Lemma \ref{lemma:beta-vk} as $\frac{1}{1-\multError} \geq 1$ for any $\multError \in [0,1)$.
    Assuming the claim is true for $1,\ldots k-1$, we show it is true for $k$:
    \begin{align*}
        &\sum_{i=1}^k \multError^{i} \valBy{k-i+1}{\retSol} = \multError \valBy{k}{\retSol} + \sum_{i=2}^k \multError^{i} \valBy{k-i+1}{\retSol}\\
        &= \multError \valBy{k}{\retSol} + \sum_{i=1}^{k-1} \multError^{i+1} \valBy{k-(i+1)+1}{\retSol} \\
        &= \multError \valBy{k}{\retSol} + \multError \sum_{i=1}^{k-1} \multError^{i} \valBy{(k-1) -i+1}{\retSol}\\
    & \equWithExp{\text{By induction assumption}}{\geq \multError \valBy{k}{\retSol} + \multError \left(\sum_{i=1}^{k-1} \valBy{i}{\retSol} - \sum_{i=1}^{k-1} z_i -\frac{1}{1-\multError} \additiveError\right)}\\
        &= \multError \sum_{i=1}^{k} \valBy{i}{\retSol} - \multError\sum_{i=1}^{k-1} z_i -\frac{\multError}{1 - \multError}  \additiveError\\
    & \equWithExp{\text{By Lemma \ref{lemma:beta-vk}}}{\geq \sum_{i=1}^{k} \valBy{i}{\retSol} -  \sum_{i=1}^{k} z_i - \additiveError -\frac{\multError}{1 - \multError}  \additiveError}\\
         &=\sum_{i=1}^k \valBy{i}{\retSol} - \sum_{i=1}^{k} z_i -\frac{1}{1 - \multError}\additiveError
    \end{align*}
        \qedhere
\end{proof}

\begin{lemma}\label{lemma:fk-to-all}
    For all $1<k \leq n$, 
    \begin{align*}
        \frac{\multError}{1-\multError} \valBy{k}{\retSol} \geq \sum_{i=1}^{k-1}\valBy{i}{\retSol} - \sum_{i=1}^{k-1}z_i - \frac{1}{1 - \multError}\additiveError
    \end{align*}
\end{lemma}

\begin{proof}
    First, notice that since $k \geq (k-1)-i+1$ for any $1\leq i \leq k$ and as the function $\valBy{i}$ represents the $i$-th smallest objective value, also:
    \begin{align}\label{eq:increase-by-obj-size}
        \forall 1\leq i \leq k \colon \Hquad \valBy{k}{\retSol} \geq \valBy{(k-1)-i+1}{\retSol}
    \end{align}
    In addition, consider the geometric series with a first element $1$, a ratio $\multError$, and a length $(k-1)$. 
    As $\multError < 1$, its sum can be bounded in the following way:
    \begin{align}\label{eq:geometric-series-beta}
        \sum_{i=1}^{k-1} \multError^{i-1} = \frac{1-\multError^{k-1}}{1-\multError} < \lim_{k \to \infty}\frac{1-\multError^{k-1}}{1-\multError} = \frac{1}{1-\multError}
    \end{align}
    
    Now, the claim can be concluded as follows:
    \begin{align*}
        & \frac{\multError}{1-\multError}\valBy{k}{\retSol} = \multError \left(\frac{1}{1-\multError} \valBy{k}{\retSol} \right)\\
        & \equWithExp{\text{By Equation \eqref{eq:geometric-series-beta}}}{ > \multError \left(\sum_{i=1}^{k-1} \multError^{i-1} \valBy{k}{\retSol} \right)}\\
        & \equWithExp{\text{By Equation \eqref{eq:increase-by-obj-size}}}{\geq  \multError \left(\sum_{i=1}^{k-1} \multError^{i-1} \valBy{(k-1)-i+1}{\retSol} \right)}\\
        &= \sum_{i=1}^{k-1} \multError^{i} \valBy{(k-1)-i+1}{\retSol} \\ 
        & \equWithExp{\text{By Lemma \ref{lemma:beta-sums-to-diff} for $(k-1)\geq 1$}}{\geq \sum_{i=1}^{k-1}\valBy{i}{\retSol} - \sum_{i=1}^{k-1}z_i -\frac{1}{1 - \multError} \additiveError}
\end{align*}
\end{proof}


We are now ready to prove the Theorem \ref{th:main}.
\begin{proof}[Proof of Theorem \ref{th:main}]
    Recall that the claim is that $\retSol$ is a $\left(\frac{\multApprox^2}{1-\multApprox + \multApprox^2}, \frac{\additiveApprox}{1-\multApprox +\multApprox^2}\right)$-leximin-approximation.
    
    For brevity, we define the following constant:
    \begin{align*}
        \Delta(\multApprox) = \frac{1}{1-\multApprox + \multApprox^2}
    \end{align*}
    Accordingly, we need to prove that $\retSol$ is a $\left(\multApprox^2\cdot\Delta(\multApprox), \additiveApprox \cdot\Delta(\multApprox)\right)$-approximation.
    
   As $\multApprox = 1 - \multError$, it is easy to verify that:
   \begin{align}\label{eq:delta-alpha-with-beta}
       \Delta(\multApprox) = \frac{1}{1-\multError + \multError^2}
   \end{align}

    Now, suppose by contradiction that $\retSol$ is \emph{not} $\left(\multApprox^2\cdot\Delta(\multApprox), \additiveApprox \cdot\Delta(\multApprox)\right)$-approximately leximin-optimal.
    By definition, this means there exists a solution $y \in S$  that is $\left(\multApprox^2\cdot\Delta(\multApprox), \additiveApprox \cdot\Delta(\multApprox)\right)$-leximin-preferred over it.
    That is, there exists an integer $1 \leq k \leq n$ such that:
    \begin{align*}
        \forall j < k \colon &\valBy{j}{y} \geq \valBy{j}{\retSol};\\
        & \valBy{k}{y} > \frac{1}{\multApprox^2\cdot\Delta(\multApprox)} \left(\valBy{k}{\retSol} + \additiveApprox \cdot\Delta(\multApprox) \right).
    \end{align*}

    Since $\retSol$ was obtained in \eqref{eq:compact-OP} that was solved in the last iteration $n$, it is clear that $\sum_{i=1}^k \valBy{i}{\retSol} \geq \sum_{i=1}^{k} z_i$ (by constraint (\progCompact.2) if $k<n$ and (\progCompact.3) otherwise).
    Which implies:
    \begin{align}\label{eq:fk-to-zk}
        \sum_{i=1}^k \valBy{i}{\retSol} - \sum_{i=1}^{k-1} z_i \geq z_k
    \end{align}

    Now, consider \eqref{eq:compact-OP} that was solved in iteration $k$.
    By Observation \ref{obs:retSol-solves-any-t}, $\retSol$ is feasible to this problem.
    As the $(k-1)$ smallest objective values of $y$ are at least as high as those of $\retSol$, it is easy to conclude that $y$ also satisfies constraints (\progCompact.2) of this problem; since, for any $\ell < k$:
    \begin{align*}
        \sum_{i=1}^{\ell} \valBy{i}{y} \geq\sum_{i=1}^{\ell} \valBy{i}{\retSol} \geq \sum_{i=1}^{\ell} z_i
    \end{align*}
    Therefore, by Observation \ref{obs:feasi-and-constraint2}, $y$ is also feasible to this problem. 

    The value obtained during the algorithm run as an approximation for this problem is $z_k$.
    This means that the optimal value is at most $\frac{1}{\multApprox}\left(z_k + \additiveError \right)$.
    As $y$ is feasible in this problem, and since the objective value obtained by $y$ in this problem is $\sum_{i=1}^k \valBy{i}{y} - \sum_{i=1}^{k-1} z_i$ (by Observation \ref{obs:obj-value}), this implies the following:
    \begin{align}\label{eq:y-upper bound} 
        \sum_{i=1}^k \valBy{i}{y} - \sum_{i=1}^{k-1} z_i \leq \frac{1}{(1-\multError)}\left(z_k+\additiveError\right)
    \end{align}

    If $k=1$, we get that the objective value obtained by $y$ is $\valBy{1}{y}$.
    In addition, $\valBy{1}{\retSol} \geq z_1$ by Equation \eqref{eq:fk-to-zk}. 
    However, as $0<\multApprox \leq 1$, then $\Delta(\multApprox) \geq 1$ but $\multApprox \cdot\Delta(\multApprox) \leq 1$.
    It follows that:
    \begin{align*}
        \valBy{1}{y}> \frac{1}{\multApprox^2 \cdot \Delta(\multApprox)} \left(\valBy{1}{\retSol} + \additiveError \cdot \Delta(\multApprox) \right)\geq \frac{1}{\multApprox} \left(z_1 + \additiveError \right)
    \end{align*}
    In contradiction to Equation \eqref{eq:y-upper bound}.

    Therefore, $k>1$.  
    We start by showing that the following holds:
    \begin{align}\label{eq:yk-to-sum}
        \valBy{k}{y} > \frac{1}{1-\multError} \left( \valBy{k}{\retSol} + \multError\sum_{i=1}^{k-1}\valBy{i}{\retSol} - \multError \sum_{i=1}^{k-1}z_i  +\additiveError \right)
    \end{align}
    Consider $\valBy{k}{y}$, by the definition of $y$ for $k$ we get that:
    \begin{align*}
        &\valBy{k}{y} > \frac{1}{ \multApprox^2 \cdot \Delta(\multApprox)} \left(\valBy{k}{\retSol} +  \additiveError \cdot \Delta(\multApprox) \right)\\
        &\equWithExp{\text{\displayComsoc{By Equ. \eqref{eq:delta-alpha-with-beta} and $\multError$'s def.}\displayEcai{Since $\multApprox = 1 - \multError$ and by Equation \eqref{eq:delta-alpha-with-beta}, this equals to}}}{= \frac{1}{\multApprox}\left(\frac{1-\multError +\multError^2}{1-\multError} \valBy{k}{\retSol}+ \frac{1}{\multApprox}\additiveError\right)}\\
        &= \frac{1}{\multApprox}\left(\valBy{k}{\retSol} + \frac{\multError^2}{1-\multError} \valBy{k}{\retSol}+ \frac{1}{\multApprox}\additiveError\right)\\
        &\equWithExp{\text{By Lemma \ref{lemma:fk-to-all}\displayEcai{ we can conclude that}}}{
        \geq \frac{1}{\multApprox}\left(\valBy{k}{\retSol} + \multError\Bigr[\sum_{i=1}^{k-1}\valBy{i}{\retSol} - \sum_{i=1}^{k-1}z_i -\frac{1}{1-\multError}\additiveError\Bigr]+ \frac{1}{\multApprox}\additiveError\right)}\\
        & \equWithExp{\text{\displayComsoc{By $\multError$'s def.}\displayEcai{As $\multApprox = 1-\multError$}}}{\frac{1}{1-\multError}\left(\valBy{k}{\retSol} + \multError\sum_{i=1}^{k-1}\valBy{i}{\retSol} -\multError \sum_{i=1}^{k-1}z_i +\additiveError\right)}
    \end{align*}    

    But, we shall now prove that this means, once again, that the objective value of $y$, which is $\sum_{i=1}^k \valBy{i}{y} - \sum_{i=1}^{k-1} z_i$, is higher than $\frac{1}{1-\multError} \left(z_k +\additiveError\right)$, in contradiction to Equation \eqref{eq:y-upper bound}:
    \begin{align*}
        &\displayComsoc{\quad}\sum_{i=1}^k \valBy{i}{y} - \sum_{i=1}^{k-1} z_i=\sum_{i=1}^{k-1} \valBy{i}{y} - \sum_{i=1}^{k-1} z_i + \valBy{k}{y}\\
        &\text{By the definition of $y$ for $i<k$:}\\
        &\displayComsoc{\quad}\geq \sum_{i=1}^{k-1} \valBy{i}{\retSol} - \sum_{i=1}^{k-1} z_i + \valBy{k}{y}\\
        &\text{By Equation \eqref{eq:yk-to-sum}:}\\
        &\displayComsoc{\quad}> \sum_{i=1}^{k-1} \valBy{i}{\retSol} - \sum_{i=1}^{k-1} z_i + \frac{1}{1-\multError} \valBy{k}{\retSol} \displayEcai{ \\
        & \quad\quad\quad} + \frac{\multError}{1-\multError}\sum_{i=1}^{k-1}\valBy{i}{\retSol} - \frac{\multError}{1-\multError}\sum_{i=1}^{k-1}z_i +\frac{1}{1-\multError}\cdot\additiveError\\
        &\text{Since  $1+\frac{\multError}{1-\multError} = \frac{1}{1-\multError}$, this equals to:}\\
        &\displayComsoc{\quad} = \frac{1}{1-\multError}\sum_{i=1}^{k-1} \valBy{i}{\retSol} - \frac{1}{1-\multError}\sum_{i=1}^{k-1} z_i + \frac{1}{1-\multError} \valBy{k}{\retSol} + \frac{1}{1-\multError} \additiveError\\
        & \displayComsoc{\quad}= \frac{1}{1-\multError} \left(\sum_{i=1}^k \valBy{k}{\retSol} - \sum_{i=1}^{k-1}z_i + \additiveError\right)\\
        &\text{By Equation \eqref{eq:fk-to-zk}:}\\
        &\displayComsoc{\quad}\geq \frac{1}{1-\multError} \left(z_k +\additiveError\right).
    \end{align*}
    This is a contradiction, so Theorem \ref{th:main} is proved.
\end{proof}



\end{document}